\documentclass[epsf]{iopart}
\usepackage{iopams,amsfonts,amsthm,rotating,graphicx,fancybox,fancyhdr,bm,eucal,dsfont,appendix}
\usepackage{datetime,esint,color}
\usepackage{accents}
\usepackage{wasysym}
\usepackage[normalem]{ulem}

\newcommand*\rfrac[2]{{}^{#1}\!/_{#2}}

\eqnobysec
\settimeformat{hhmmsstime}
\ddmmyyyydate

\theoremstyle{plain}
\newtheorem{theorem}{Theorem}[section]
\newtheorem{lemma}[theorem]{Lemma}
\newtheorem{proposition}[theorem]{Proposition}
\newtheorem{corollary}[theorem]{Corollary}
\theoremstyle{definition}
\newtheorem{definition}[theorem]{Definition}
\newtheorem{remark}[theorem]{Remark}

\newcommand{\CC}{{\mathbb C}}
\newcommand{\RR}{{\mathbb R}}

\begin{document}

\title{Power spectrum and form factor in random diagonal matrices and integrable billiards}

\author{\hspace{-2cm}Roman Riser$^{1,2}$ and Eugene Kanzieper$^1$\footnote[1]{Corresponding author.}}

\address{\hspace{-2cm}$^1$~Department of Mathematics, Holon Institute of Technology, Holon 5810201, Israel\\
         \hspace{-2cm}$^2$~Department of Mathematics and Research Center for Theoretical Physics\\
         \hspace{-1.8cm} and Astrophysics, University of Haifa, Haifa 3498838, Israel
        }
\noindent\newline\newline
\begin{abstract}
Triggered by a controversy surrounding a universal behaviour of the power spectrum in quantum systems exhibiting regular classical dynamics,
we focus on a model of random diagonal matrices (RDM), often associated with the Poisson spectral universality class, and examine how the power spectrum and the form factor
get affected by two-sided truncations of RDM spectra. Having developed a nonperturbative description of both statistics, we perform their detailed asymptotic analysis to demonstrate explicitly how a traditional assumption (lying at the heart of the controversy) -- that the power spectrum is merely determined by the spectral form factor -- breaks down for truncated spectra. This observation has important consequences as we further argue that bounded quantum systems with integrable classical dynamics are described by heavily truncated rather than complete RDM spectra. High-precision numerical simulations of semicircular and irrational rectangular billiards lend independent support to these conclusions.
\end{abstract}
\newpage
\tableofcontents

\newpage
\section{Introduction}\label{Intro}

Since the second half of the XX century, a variety of statistical indicators have been devised to study fluctuations in spectra of bounded quantum systems, at both short- and long-range energy scales. For high-lying energy levels, these indicators, or spectral fluctuation measures, were found to exhibit an {\it identical} behavior in {\it different} quantum systems, thereby revealing the phenomenon of {\it spectral universality}~\cite{B-1987}. To observe it, an influence of system-specific mean level density has to be eliminated from sequences of measured energy levels $\{ \lambda_1,\dots,\lambda_N\}$. This is achieved by means of the unfolding procedure~\cite{M-2004} represented by the map
\begin{eqnarray} \label{unfolding_map}
    \varepsilon_\ell = \int_{-\infty}^{\lambda_\ell} d\lambda \varrho_N(\lambda),\quad \ell=1,\dots,N,
\end{eqnarray}
where $\varrho_N(\lambda)$ is the mean level density of the original (e.g.,~measured) spectrum. It is the unfolded energy levels $\{ \varepsilon_1,\dots,\varepsilon_N\}$ that may obey universal statistical laws as $N \rightarrow \infty$.

There is a broad consensus, based on a vast amount of experimental, numerical and theoretical evidence~\cite{GMGW-1998,S-1999,H-2001}, that a universal statistical behavior of quantum systems correlates with the nature -- chaotic or regular -- of their underlying classical dynamics. This gives rise to two paradigmatic universality classes in quantum chaology.

The {\it Wigner-Dyson} universality class accommodates generic quantum systems which are {\it fully chaotic} in the classical limit. According to the Bohigas-Giannoni-Schmit (BGS) conjecture \cite{BGS-1984}, spectral fluctuations of highly excited energy levels in such systems -- exhibiting long-range correlations
and local repulsion -- are governed by global symmetries rather than by system specialties and are accurately described by the random matrix theory~\cite{M-2004,PF-book}. Mostly built on numerical indications~\cite{MK-1979,CVG-1980,B-1981,BGS-1984}, the BGS conjecture has later been advocated within the field-theoretic~\cite{AASA-1996} and semiclassical~\cite{RS-2002} approaches. On the contrary, generic quantum systems whose classical dynamics is {\it completely integrable} belong to a different -- {\it Poisson} -- universality class as was first conjectured by Berry and Tabor \cite{BT-1977}. Spectral fluctuations therein are radically different from those in Wigner-Dyson spectra, with energy levels being completely uncorrelated.

Since the most commonly used spectral fluctuation measures~\cite{M-2004,PF-book} -- be it the distribution of level spacings~\cite{JMMS-1980} or that of the ratio of two consecutive level spacings~\cite{ABGR-2013}, the number variance~\cite{B-1988} or spectral rigidity~\cite{B-1985} -- probe spectral correlations on {\it either} local {\it or} global scales, a combination of {\it both} short- and long-range statistical indicators is often required to get a reliable information on the degree or type of chaoticity of a quantum system in question. This remark is even more pertinent, if not crucial, for analysis of quantum systems with mixed classical dynamics~\cite{GRRFSVR-2005} and for an accurate interpretation of experimental data in which the completeness of energy spectra is rather a rare situation~\cite{FKMMRR-2006,LBYBS-2018}.

In this context, we would like to draw reader's attention to an alternative measure of spectral fluctuations -- the {\it power spectrum} of eigenlevel sequences -- which is capable of probing the correlations between {\it both} nearby and distant eigenlevels. Introduced long ago by Odlyzko in his famous study~\cite{Od-1987} of the spacing distribution between nontrivial zeros of the Riemann zeta function, and re-invented fifteen years later by Rela\~{n}o and collaborators~\cite{RGMRF-2002}, it considers a sequence of discrete energy levels in a quantum system as a discrete time series, where indices of {\it ordered} eigenlevels play the role of discrete time.

\begin{definition}\label{def-01}
  Let $\{\varepsilon_1 \le \dots \le \varepsilon_N\}$ be a sequence of ordered unfolded eigenlevels, $N \in {\mathbb N}$, with the mean level spacing $\Delta$ and let
  $\langle \delta\varepsilon_\ell \delta\varepsilon_m \rangle$ be the covariance matrix of level displacements $\delta\varepsilon_\ell = \varepsilon_\ell - \langle \varepsilon_\ell\rangle$ from their mean $\langle \varepsilon_\ell\rangle=\ell \Delta$. A Fourier transform of the covariance matrix
\begin{eqnarray}\label{ps-def}
    S_N(\omega) = \frac{1}{N\Delta^2} \sum_{\ell=1}^N \sum_{m=1}^N \langle \delta\varepsilon_\ell \delta\varepsilon_m \rangle\, e^{i\omega (\ell-m)}, \quad \omega \in {\mathbb R}
\end{eqnarray}
is called the power spectrum of a sequence. Here, the angular brackets stand for an average over an ensemble of eigenlevel sequences.
\hfill $\blacksquare$
\end{definition}

\begin{remark}
Since the power spectrum is $2\pi$-periodic, real and even function in $\omega$,
\begin{eqnarray}\fl \qquad
S_N(\omega+2\pi) = S_N(\omega), \quad S_N^*(\omega) = S_N(\omega), \quad S_N(-\omega) = S_N(\omega),
\end{eqnarray}
it is sufficient to consider it in the interval $0 \le \omega \le \omega_{\rm Ny}$, where $\omega_{\rm Ny} = \pi$ is the Nyquist frequency.
\hfill $\blacksquare$
\end{remark}

\begin{remark}
It readily follows from Definition~\ref{def-01} that, by tuning the frequency $\omega$ in the power spectrum, one may attend to spectral correlations between either adjacent or distant eigenlevels. Indeed, at `large' frequencies, $\omega={\mathcal O}(N^0)$ yet below the Nyquist frequency, the distant eigenlevels barely contribute to the power spectrum; characterized by large values of $|\ell-m|$, they produce strongly oscillating terms in Eq.~(\ref{ps-def}) which effectively cancel each other. As the result, $S_N(\omega)$ is mainly shaped by correlations between the nearby levels. At low frequencies, $\omega \ll 1$, these oscillations are by far less pronounced thus making a contribution of distant eigenlevels increasingly important.
\hfill $\blacksquare$
\end{remark}

Although the power spectrum, as an alternative spectral statistics, was suggested quite some time ago~\cite{Od-1987,RGMRF-2002}, its {\it nonperturbative} theory~\cite{ROK-2017,ROK-2020} was formulated only recently. (For an earlier, heuristic approach exploiting a form factor approximation, see Ref.~\cite{FGMMRR-2004}.) In particular, in Ref.~\cite{ROK-2017}, the power spectrum of energy level fluctuations in {\it fully chaotic} quantum structures with broken time-reversal symmetry was determined by employing a random-matrix-theory approach. A universal, parameter-free prediction for the power spectrum expressed in terms of a fifth Painlev\'e transcendent was the main result of that study. The ideas underlying a random-matrix-theory calculation of Ref.~\cite{ROK-2017} were further developed in Ref.~\cite{ROK-2020}, where the power spectrum was studied in a more general setting which is applicable to a much broader class of quantum systems whose eigenspectra -- {\it not necessarily of the random-matrix-theory type} -- possess stationary level spacings.

In this note, we utilize a nonperturbative approach of Ref.~\cite{ROK-2020} to determine the power spectrum for generic quantum systems with {\it integrable} classical dynamics~\footnote{Recent report~\cite{PPK-2020} makes our previous~\cite{ROK-2020} and present study potentially relevant to the problem of universal correlations in spectra of quantum {\it many-body} localized systems. Notice that the main analytical result of Ref.~\cite{PPK-2020}, see their Eqs.~(8) and (9), is a particular case of the spectral form factor derived in Ref.~\cite{ROK-2020} for a model of eigenlevel sequences with identically distributed, uncorrelated level spacings, see Eqs.~(1.21) and (1.22) therein. Indeed, setting the characteristic function $\Psi_s(\tau)$ of level spacings of Ref.~\cite{ROK-2020} to $\Psi_s(\tau) = (1-2i\pi\tau)^{-1}$ describing exponentially fluctuating spacings, one obtains the result of Ref.~\cite{PPK-2020}.}. We argue that an adequate description of such systems, modelled by random diagonal matrices~\cite{G-1996}, must take into account the effect of {\it eigenspectra truncations} which are unavoidable in any realistic experimental and numerical setup.

The paper is organized as follows. In Section~\ref{PSSF}, we introduce the notions of random spectra with stationary~\cite{ROK-2020} and infinitely-stationary level spacings, and meticulously examine how truncations of such spectra affect the power spectrum and the form factor of a quantum system. Our analysis reveals that a much sought-after information about {\it individual} fluctuations of eigenlevels in truncated sequences gets hindered by their {\it collective} fluctuations. This observation leads us to formulate meaningful definitions of the power spectrum (Definition~\ref{def-01-truncated}) and of the form factor (Definition~\ref{def-01-K-tau-truncated}) for truncated eigenlevel sequences. Theorem~\ref{PS-stationary-truncated} for the power spectrum of truncated spectra is the main result of Section~\ref{PSSF}. Our treatment of both statistics here is not restricted to a particular spectral model or universality class; yet, it will later be used to address the issue of validity of the form factor approximation in the context of quantum systems with integrable classical dynamics.

In Section~\ref{rdm-section}, a general framework of Section~\ref{PSSF} is applied to a model of random diagonal matrices (RDM) often used to mimic spectral fluctuations in quantum systems exhibiting regular classical dynamics, see e.g. Ref.~\cite{G-1996}.

\begin{figure}
\includegraphics[width=\textwidth]{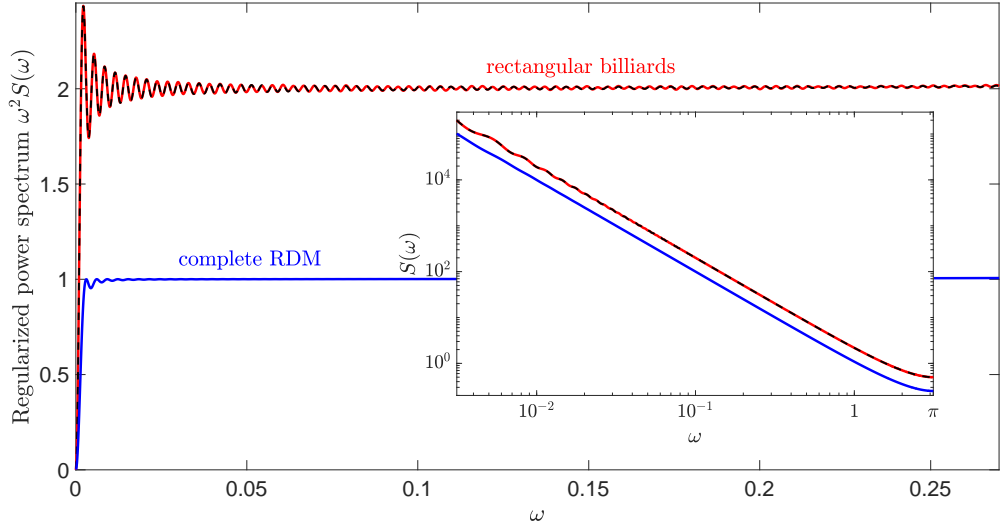}
\caption{Comparison of the regularized power spectrum for quantum irrational rectangular billiards (red solid line, numerical simulation, Section~\ref{IRBN}) with that for complete (blue solid line) and heavily truncated (black dashed line) spectrum of random diagonal matrices (theory, Section~\ref{SFF-sec}). {\it The power spectrum in rectangular billiards is seen to follow predictions for heavily truncated random diagonal matrices.} Oscillatory behavior of two upper curves is a fingerprint of inevitable finite-size effects originating from a fixed length $L$ of `measured' spectral sequences; in all cases, sequences with $L=2048$ were used. Inset: a log-log plot for the power spectra shown up to the Nyquist frequency $\omega_{\rm Ny}=\pi$; the same color legend is used.
}
\label{Plot_1}
\end{figure}

In Section~\ref{RDM-unf}, we introduce a procedure of spectral unfolding in the finite--$N$ random diagonal matrices, show that their ordered unfolded eigenlevels are beta-distributed (Proposition~\ref{prop-JPDF-ordered}), and prove infinite-stationarity of level spacings on the unfolded scale (Proposition~\ref{prop-inf}). These results hold true for RDM of arbitrary size $N$.

In Section~\ref{PS-RDM}, we evaluate the power spectrum of complete and truncated eigenlevel sequences drawn from an unfolded RDM spectrum. Both exact (Theorem~\ref{SNL-theorem}, finite--$N$) and asymptotic (Corollaries~\ref{PS-inf-L-omega} and \ref{PS-SL-trunc}, $N\rightarrow \infty$) expressions for the power spectrum constitute the main results of Section~\ref{PS-RDM}. Their implications are further discussed in the context of a quantitative description of the universal spectral fluctuations in generic quantum systems with regular classical geodesics. In paticular, we argue that heavily truncated -- rather than complete -- spectra of random diagonal matrices should be employed to adequately determine the power spectrum in realistic quantum systems. This leads us to re-confirm our previous result~\cite{ROK-2017,ROK-2020}
\begin{eqnarray}\label{main-01}
    S_{\rm int}(\omega) = \frac{1}{2} \times \frac{1}{\sin^2(\omega/2)}, \quad \omega_0 < \omega \le \pi
\end{eqnarray}
for the power spectrum of generic `integrable' quantum systems, away from the origin $\omega=0$ (here $\omega_0>0$ is an arbitrary small real number). Notice the factor $\rfrac{1}{2}$ which differs from the factor $\rfrac{1}{4}$ erroneously claimed in the earlier study~\cite{FGMMRR-2004} which used a so-called form factor approximation. We stress that the parameter free prediction Eq.~(\ref{main-01}) is a universal limit of a more general expression for the power spectrum [Eq.~(\ref{SL-omega}) of Corollary~\ref{PS-inf-L-omega}] which accounts for finite-size effects arising due to a finite length $L$ of `measured' spectral sequences. For readers' convenience, these conclusions are visualized in Fig.~\ref{Plot_1}, where the power spectrum of quantum irrational rectangular billiards is seen to follow the `$\rfrac{1}{2}$' (rather than
`$\rfrac{1}{4}$') law claimed in Eq.~(\ref{main-01}).

In Section~\ref{SFF-sec}, we focus on a nonperturbative analysis of the spectral form factor for complete and truncated eigenlevel sequences sampled from an unfolded, finite--$N$ RDM spectrum. The exact solution for the spectral form factor (Theorem~\ref{KNL-theorem}) is further analysed asymptotically (Theorem~\ref{KNL-infinite}) to pave the way for an in-depth discussion of a status of the form factor approximation. Theorem~\ref{KNL-theorem} and Theorem~\ref{KNL-infinite} are the main results of Section~\ref{SFF-sec}.

Such a discussion is held in Section~\ref{FFA-status}, where we compare a large--$N$ behavior of the spectral form factor and the power spectrum in various time and frequency scaling limits, in an attempt to scrutinize a heuristic relation Eq.~(\ref{FFA-heuristics}) claimed~\cite{FGMMRR-2004} to link the power spectrum to the spectral form factor. While we show that this relation does hold true for {\it complete} spectra of {\it infinite-dimensional random diagonal matrices} up to $\omega={\mathcal O}(N^0)$ as soon as $\omega \ll 1$, it becomes invalid for {\it truncated} spectra (of the length $L<N$). In the latter case, the validity of the form factor approximation is restricted to very low frequencies $\omega={\mathcal O}(L^{-1})$; it breaks down for $\omega={\mathcal O}(L^{-1/2})$ and certainly gives no access to the bulk of the power spectrum.

In Section~\ref{Num-Billiards}, we present a thorough numerical study of the power spectrum and the spectral form factor for quantum (i) semicircular and (ii) irrational rectangular billiards. The numerics lends unequivocal support to our theory which associates generic quantum systems possessing integrable classical dynamics with heavily truncated (rather than complete) spectra of random diagonal matrices.

Appendix~\ref{A-1} contains proofs of various asymptotic formulae for the Kummer confluent hypergeometric function ${}_1 F_1$ of imaginary argument, required for the large--$N$ analysis of the RDM spectral form factor, Section~\ref{SFF-sec}. To the best of our knowledge, explicit asymptotic results obtained in Appendix~\ref{A-1} are unavailable in the existing literature on special functions and thus can be of independent interest.

\section{Power spectrum and spectral form factor for truncated spectra with stationary spacings} \label{PSSF}
\fancyhead{} \fancyhead[RE,LO]{\textsl{\thesection \quad Power spectrum and spectral form factor for truncated spectra}}
\fancyhead[LE,RO]{\thepage}
\subsection{Random spectra with stationary and infinitely-stationary spacings}
To set up the scene, we first collect the most useful definitions and statements regarding random spectra with {\it stationary} and {\it infinitely-stationary} level spacings~\cite{ROK-2020}.

\begin{definition}[Stationarity]\label{def-stationary} Consider an ordered sequence of eigenlevels $\{0 \le \varepsilon_1 \le \dots \le \varepsilon_N\}$ with $N \in \mathbb{N}$. Let $\{s_1, \dots, s_N \}$ be the sequence of spacings between consecutive eigenlevels such that $s_\ell = \varepsilon_\ell - \varepsilon_{\ell-1}$ with $\ell=1,\dots,N$ and $\varepsilon_0=0$. The sequence of level spacings is said to be {\it stationary} if (i) the average spacing
\begin{eqnarray}\label{skd}
\langle s_\ell \rangle = \Delta
\end{eqnarray}
is independent of $\ell=1,\dots,N$ and (ii) the covariance matrix of {\it spacings} is of the Toeplitz type:
\begin{eqnarray}\label{toep}
    {\rm cov}(s_\ell, s_m) =  I_{|\ell-m|} - \Delta^2
\end{eqnarray}
for all $\ell,m=1,\dots,N$. Here, $I_n$ is a function defined for non-negative integers $n$.
\hfill $\blacksquare$
\end{definition}
\noindent
A necessary and sufficient condition for eigenlevel sequences to possess stationary level spacings, proven in Section 3 of Ref.~\cite{ROK-2020}, is formulated in Lemma \ref{Lemma-stat}.

\begin{lemma} \label{Lemma-stat} For $N \in {\mathbb N}$, let $\{0\le \varepsilon_1 \le \dots \le \varepsilon_N\}$ be an ordered sequence of unfolded eigenlevels such that $\langle
\varepsilon_1\rangle =\Delta$. An associated sequence of spacings between consecutive eigenlevels is \textbf{\textit{stationary}} if and only if
\begin{eqnarray}\label{L-1}
    \langle (\varepsilon_\ell - \varepsilon_m)^q\rangle = \langle \varepsilon_{\ell-m}^q \rangle
\end{eqnarray}
for $\ell>m$ and both $q=1$ and $q=2$.
\end{lemma}

\begin{remark}
Examples of {\it finite}--$N$ eigenlevel sequences with stationary spacings include (i) random spectra with independent, identically distributed (i.i.d.)~level spacings~\footnote[5]{As the $\ell$-th ordered eigenlevel is a sum of $\ell$ i.i.d. random variables, $\varepsilon_\ell = \sum_{j=1}^\ell s_j$, both the l.h.s. and r.h.s. in Eq.~(\ref{L-1}) represent the $q$-th moment of a sum of $(\ell-m)$ i.i.d. random variables, $q \in {\mathbb N}$. (Infinite) stationarity follows immediately.}; (ii) unfolded spectra in ensembles of random diagonal matrices (see Section~\ref{rdm-section}); (iii) spectra of `tuned' circular ensembles of random matrices, see Section 4 in Ref.~\cite{ROK-2020}. Interestingly, spectral fluctuations in all aforementioned models appear to satisfy the correlation constraint Eq.~(\ref{L-1}) also for $q>2$. This leads us to define the notion of random spectra with infinitly-stationary level spacings.
\hfill $\blacksquare$
\end{remark}

\begin{definition}[Infinite stationarity]\label{def-inf-stationary} An ordered sequence of eigenlevels $\{0 \le \varepsilon_1 \le \dots \le \varepsilon_N\}$ with $N \in \mathbb{N}$ is said
to possess an \textbf{\textit{infinitly-stationary}} sequence of level spacings if the relation
\begin{eqnarray}\label{L-1-inf}
    \langle (\varepsilon_\ell - \varepsilon_m)^q\rangle = \langle \varepsilon_{\ell-m}^q \rangle
\end{eqnarray}
holds for $\ell>m$ and {\it all} $q \in {\mathbb N}$.
\hfill $\blacksquare$
\end{definition}

\subsection{Power spectrum for truncated eigenlevel sequences with stationary spacings} \label{PS-trunc-ss}

For {\it complete} (untruncated) sequences of eigenlevels with {\it stationary} level spacings, the power spectrum can be represented as a differential operator acting on a generating function of eigenlevel variances, see Theorem 2.3 in Ref.~\cite{ROK-2020}:

\begin{theorem}\label{PS-stationary-main}
  Let $N \in \mathbb{N}$ and $0\le \omega \le \pi$.  The power spectrum for an eigenlevel sequence $\{0\le \varepsilon_1 \le \dots \le \varepsilon_N\}$ with stationary spacings equals
\begin{eqnarray}\label{smd-sum}
    S_N(\omega) = \frac{1}{N \Delta^2} {\rm Re} \left(  z \frac{\partial}{\partial z} - N - \frac{1-z^{-N}}{1-z}\right)
        \sum_{\ell=1}^N {\rm var}[\varepsilon_\ell]\, z^\ell,
\end{eqnarray}
where $z=e^{i \omega}$, $\Delta$ is the mean level spacing, and
\begin{eqnarray} \label{var-N}
{\rm var}[\varepsilon_\ell] = \langle \delta\varepsilon_\ell^2 \rangle
\end{eqnarray}
is the variance of $\ell$-th ordered eigenlevel.
\end{theorem}

\begin{remark}
Theorem~\ref{PS-stationary-main}, as it stands, cannot directly be applied to describe the power spectrum in a realistic quantum system with unbounded spectrum for two reasons. First, in an experiment one normally measures a {\it truncated spectrum} represented by a subsequence~\footnote[6]{The eigenlevels are assumed to be unfolded.} $\{ \varepsilon_{M+1} \le \dots \le \varepsilon_{M+L}\}$ of, say, $L$ eigenlevels obtained from a complete spectrum $\{0 \le \varepsilon_{1} \le \dots \le \varepsilon_{N}\}$ of the length $N$ by skipping both $M$ low-lying eigenlevels $\{\varepsilon_1 \le \dots \le \varepsilon_{M}\}$ and $N-M-L$ high-lying eigenlevels $\{ \varepsilon_{M+L+1} \le \dots \le \varepsilon_{N}\}$; here, $M+L \le N$. Second, obviously, the $N\rightarrow \infty$ limit should be implemented. While the latter problem is merely technical, the former one raises the question of how the original Definition~\ref{def-01} of the power spectrum should be modified to account for {\it truncated} eigenlevel sequences in a meaningful way.\newline
${}$\hfill $\blacksquare$
\end{remark}
\noindent
{\it Power spectrum of truncated spectrum and collective fluctuations.}---To characterize the power spectrum of an ordered subsequence $\{ \varepsilon_{M+1} \le \dots \le \varepsilon_{M+L}\}$ of the complete spectrum $\{ 0 \le \varepsilon_1 \le \dots \le \varepsilon_{N}\}$, one might think of a simple-minded extension of Definition~\ref{def-01} in the form
\begin{eqnarray}\label{ps-def-trial}
    \tilde{S}_{N; M, L}(\omega) = \frac{1}{L\Delta^2} \sum_{\ell=M+1}^{M+L} \sum_{m=M+1}^{M+L} \langle \delta\varepsilon_\ell \delta\varepsilon_m \rangle\, e^{i\omega (\ell-m)}.
\end{eqnarray}
Unfortunately, there is a serious drawback to this proposal. In order to reveal it, we assume stationarity of eigenlevel spacings [Eq.~(\ref{L-1}) at $q=2$] which brings the relation
\begin{eqnarray} \label{delta-2}
    \langle \delta\varepsilon_\ell  \delta\varepsilon_m\rangle = \frac{1}{2} \left(
        \langle \delta\varepsilon_\ell^2 \rangle + \langle \delta\varepsilon_m^2 \rangle - \langle \delta\varepsilon_{|\ell-m|}^2 \rangle
        \right),
\end{eqnarray}
where $\delta\varepsilon_\ell = \varepsilon_\ell -\ell \Delta$. Substituting Eq.~(\ref{delta-2}) into the trial definition Eq.~(\ref{ps-def-trial}), we
reduce the latter to
\begin{eqnarray}\label{PS-x-01} \fl \qquad
    \tilde{S}_{N; M, L}(\omega) =  \frac{1}{2 L\Delta^2} \Bigg\{ 2 {\rm Re} \sum_{\ell,m=M+1}^{M+L}
        \langle \delta \varepsilon_\ell^2 \rangle \, z^{\ell-m}
        -
        \sum_{\ell,m=M+1}^{M+L}
        \langle \delta \varepsilon_{|\ell-m|}^2 \rangle \, z^{\ell-m}
        \Bigg\}.
\end{eqnarray}
Both terms in Eq.~(\ref{PS-x-01}) can be simplified by observing that
\begin{eqnarray} \label{1st_sum_res}
    \sum_{\ell,m=M+1}^{M+L}
        \langle \delta \varepsilon_\ell^2 \rangle \, z^{\ell-m} =
        \frac{1-z^{-L}}{z-1} \sum_{\ell=M+1}^{M+L} \langle \delta \varepsilon_\ell^2\rangle z^{\ell-M}
\end{eqnarray}
and
\begin{eqnarray}\label{2nd_sum_res}
    \sum_{\ell,m=M+1}^{M+L}
        \langle \delta \varepsilon_{|\ell-m|}^2 \rangle \, z^{\ell-m} = 2{\rm Re}\, \sum_{\ell=1}^{L}
        (L-\ell) \langle \delta \varepsilon_{\ell}^2 \rangle \, z^\ell.
\end{eqnarray}
Combining Eqs.~(\ref{PS-x-01}), (\ref{1st_sum_res}) and (\ref{2nd_sum_res}), we manage to rewrite Eq.~(\ref{ps-def-trial}) in a more suggestive form~\footnote[7]{Notice that the variance ${\rm var}[\varepsilon_\ell]$ in Eq.~(\ref{SML-new}) and below may depend on the length $N$ of the complete spectrum.}:
\begin{eqnarray} \label{SML-new} \fl
    \quad  \tilde{S}_{N; M,L}(\omega) = \frac{1}{L \Delta^2} {\rm Re} \left(
        z\frac{\partial}{\partial z} - L - \frac{1-z^{-L}}{1-z}
    \right) \sum_{\ell=1}^{L} {\rm var}[\varepsilon_\ell] z^\ell + \delta \tilde{S}_{N;M,L}(\omega),
\end{eqnarray}
where
\begin{eqnarray} \label{d-SML-new}\fl
    \qquad
    \delta\tilde{S}_{N;M,L}(\omega)
    = \frac{1}{L \Delta^2} {\rm Re} \left(
        \frac{1-z^{-L}}{1-z} \sum_{\ell=1}^{L} \left(
            {\rm var}[\varepsilon_\ell] - {\rm var}[\varepsilon_{\ell+M}]
        \right)\, z^\ell
    \right).
\end{eqnarray}
Comparing the first term in Eq.~(\ref{SML-new}) with the one in Eq.~(\ref{smd-sum}), one realizes that it does {\it not} depend on the number $M$ of dropped low-lying eigenlevels, effectively describing the power spectrum of the first $L$ eigenlevels as if $M=0$. The second term in Eq.~(\ref{SML-new}), specified explicitly in Eq.~(\ref{d-SML-new}), {\it does} depend on $M$. Nullifying at $M=0$, it represents a {\it parasitic signal}, which may dramatically mask a contribution of the first term in Eq.~(\ref{SML-new}).

This is best seen in the model of eigenlevel sequences with identically distributed uncorrelated spacings~\cite{ROK-2020}, characterized by the variance ${\rm var}[\varepsilon_\ell]= \sigma^2\ell$ linearly depending on the number of $\ell$-th ordered eigenlevel (here $\sigma^2$ is the variance of level spacings). Setting the mean level spacing to unity, $\Delta=1$, we observe that the second term in Eq.~(\ref{SML-new}), see Eq.~(\ref{d-SML-new}), reduces to
\begin{eqnarray} \label{SML-new-2}
        \delta\tilde{S}_{N;M,L}(\omega) =  \sigma^2 \frac{M}{L} \, \frac{\sin^2(\omega L/2)}{\sin^2(\omega/2)}.
\end{eqnarray}
For $M \gg L \gg 1$, it produces a giant, wildly oscillating contribution -- a parasitic signal -- whose amplitude may far exceed that of the first term.

An $M$-dependent oscillatory term $\delta\tilde{S}_{N;M,L}(\omega)$ in Eq.~(\ref{SML-new}) originates from fluctuations of the first $M$ low-lying eigenlevels $\{ 0 \le \varepsilon_1 \le \dots \le \varepsilon_{M}\}$ which, albeit skipped, {\it do} contribute to {\it collective} fluctuations of ordered eigenlevels in the subsequence $\{ \varepsilon_{M+1} \le \dots \le \varepsilon_{M+L}\}$ of our interest.
\newline\noindent\newline
{\it Eliminating collective fluctuations.}---To get rid of collective fluctuations, one has to modify the trial definition Eq.~(\ref{ps-def-trial}) as follows:

\begin{definition}\label{def-01-truncated}
  Let $\{0 \le \varepsilon_1 \le \dots \le \varepsilon_N\}$ be a sequence of ordered unfolded eigenlevels, $N \in {\mathbb N}$, with the mean level spacing $\Delta$ and let
  $\{ \varepsilon_{M+1} \le \dots \le \varepsilon_{M+L}\}$ be a spectral subsequence of the length $L \in {\mathbb N}$ obtained from the complete sequence by omitting both $M\in {\mathbb N} \cup \{0\}$ low-lying eigenlevels $\{\varepsilon_1 \le \dots \le \varepsilon_{M}\}$ and $N-M-L$ high-lying eigenlevels $\{ \varepsilon_{M+L+1} \le \dots \le \varepsilon_{N}\}$, such that $M+L\le N$. The power spectrum of the subsequence $\{ \varepsilon_{M+1} \le \dots \le \varepsilon_{M+L}\}$ equals
\begin{eqnarray}\label{ps-def-truncated} \fl
    S_{N; M, L}(\omega) = \frac{1}{L\Delta^2} \sum_{\ell=M+1}^{M+L} \sum_{m=M+1}^{M+L} \langle (\delta\varepsilon_\ell - \delta\varepsilon_M)
                        (\delta\varepsilon_m -  \delta\varepsilon_M) \rangle\, e^{i\omega (\ell-m)}, \quad \omega \in {\mathbb R},
\end{eqnarray}
where $\delta\varepsilon_\ell = \varepsilon_\ell - \langle \varepsilon_\ell\rangle$ is a displacement of the $\ell$-th ordered eigenlevel $\varepsilon_\ell$ from its mean $\langle \varepsilon_\ell\rangle$; for $M=0$ (no truncation from below), we set $\delta\varepsilon_0=0$.
\hfill $\blacksquare$
\end{definition}

\begin{remark}\label{rem-collective}
  Definition~\ref{def-01-truncated} is essentially the power spectrum of a `collectively tuned' subsequence $\{ \varepsilon_{1}^\prime, \dots, \varepsilon_{L}^\prime\}$, where
  $\varepsilon_j^\prime = \varepsilon_{M+j}-\varepsilon_{M}$. Notice that stationarity of the complete spectrum implies stationarity of a `collectively tuned' subsequence, in the sense of Definition~\ref{def-stationary}.
\hfill $\blacksquare$
\end{remark}

\begin{theorem}\label{PS-stationary-truncated}
  The power spectrum of an ordered subsequence $\{ \varepsilon_{M+1} \le \dots \le \varepsilon_{M+L}\}$ of the complete spectrum $\{ 0 \le \varepsilon_1 \le \dots \le \varepsilon_{N}\}$ with stationary spacings equals
\begin{eqnarray}\label{smd-sum-x}
    S_{N;M,L}(\omega) = \frac{1}{L \Delta^2} {\rm Re} \left(  z \frac{\partial}{\partial z} - L - \frac{1-z^{-L}}{1-z}\right)
        \sum_{\ell=1}^L {\rm var}[\varepsilon_\ell]\, z^\ell,
\end{eqnarray}
where $z=e^{i \omega}$, $\Delta$ is the mean level spacing, and
\begin{eqnarray}
{\rm var}[\varepsilon_\ell] = \langle \delta\varepsilon_\ell^2 \rangle
\end{eqnarray}
is the variance of $\ell$-th ordered eigenlevel.
\end{theorem}
\noindent
\begin{proof}
Simple rewriting of Eq.~(\ref{ps-def-truncated}) yields:
\begin{eqnarray} \label{PS-x-02} \fl \qquad
    S_{N;M,L}(\omega) = \tilde{S}_{N; M,L}(\omega) - \frac{2}{L\Delta^2} {\rm Re\,} \sum_{\ell,m=M+1}^{M+L} \langle \delta \varepsilon_\ell \delta \varepsilon_M \rangle z^{\ell-m}
    \nonumber\\ \fl \qquad \qquad\qquad \qquad\qquad
     +  \frac{1}{L\Delta^2} \sum_{\ell,m=M+1}^{M+L} \langle \delta \varepsilon_M^2 \rangle z^{\ell-m}.
\end{eqnarray}
Here, $\tilde{S}_{N; M,L}(\omega)$ is defined by Eq.~(\ref{ps-def-trial}). Making use of stationarity, Eq.~(\ref{delta-2}) with $m=M$, one has:
\begin{eqnarray} \label{PS-x-03}\fl \qquad
    S_{N;M,L}(\omega) = \tilde{S}_{N; M,L}(\omega) - \frac{1}{L\Delta^2} {\rm Re\,} \sum_{\ell,m=M+1}^{M+L}
    \Big(
        \langle \delta \varepsilon_\ell^2 \rangle - \langle \delta \varepsilon_{\ell -M}^2 \rangle
    \Big) z^{\ell -m}.
\end{eqnarray}
Performing summation over $m$, one obtains after some algebra:
\begin{eqnarray} \label{PS-x-04}
    S_{N;M,L}(\omega) = \tilde{S}_{N; M,L}(\omega) - \delta \tilde{S}_{N; M,L}(\omega),
\end{eqnarray}
see Eqs.~(\ref{SML-new}) and (\ref{d-SML-new}) for explicit expressions. Combining them with Eq.~(\ref{PS-x-04}) we derive Eq.~(\ref{smd-sum-x}).
\end{proof}

\begin{remark}\label{FF-rem}
Since the power spectrum Eq.~(\ref{ps-def-truncated}) of an ordered subsequence $\{ \varepsilon_{M+1} \le \dots \le \varepsilon_{M+L}\}$ drawn from the {\it stationary} spectrum does {\it not} depend on the number
$M$ of skipped low-lying eigenvalues, see Eq.~(\ref{smd-sum-x}), from now on we shall use the notation $S_{N;L}(\omega)$ instead of $S_{N;M,L}(\omega)$.
\hfill $\blacksquare$
\end{remark}

\subsection{Spectral form factor for truncated eigenlevel sequences with stationary spacings}\label{SFF-trunc}

Prompted by the above discussion, we define the form factor of {\it truncated} spectrum as follows:

\begin{definition}\label{def-01-K-tau-truncated}
  Let $\{0 \le \varepsilon_1 \le \dots \le \varepsilon_N\}$ be a sequence of ordered unfolded eigenlevels, $N \in {\mathbb N}$, with the mean level spacing $\Delta$ and let
  $\{ \varepsilon_{M+1} \le \dots \le \varepsilon_{M+L}\}$ be a spectral subsequence of the length $L \in {\mathbb N}$ obtained from the complete sequence by omitting both $M\in {\mathbb N} \cup \{0\}$ low-lying eigenlevels $\{\varepsilon_1 \le \dots \le \varepsilon_{M}\}$ and $N-M-L$ high-lying eigenlevels $\{ \varepsilon_{M+L+1} \le \dots \le \varepsilon_{N}\}$, such that $M+L\le N$. The form factor of the truncated spectrum equals
\begin{eqnarray}\label{ps-def-K-tau-truncated} \fl
    \qquad K_{N; M, L}(\tau) &=& \frac{1}{L}
    \Bigg( \Big<
        \sum_{\ell=M+1}^{M+L} \sum_{m=M+1}^{M+L} e^{2 i \pi \tau (\varepsilon_\ell - \varepsilon_m)/\Delta }
    \Big> \nonumber\\
    \fl \qquad
    &-&  \Big<
        \sum_{\ell=M+1}^{M+L} e^{2 i \pi \tau (\varepsilon_\ell-\varepsilon_M)/\Delta }
    \Big>
    \Big<
        \sum_{m=M+1}^{M+L} e^{-2 i \pi \tau (\varepsilon_m-\varepsilon_M)/\Delta }
    \Big> \Bigg).
\end{eqnarray}
\hfill $\blacksquare$
\end{definition}
\begin{remark}\label{rem-collective-sff}
  Definition~\ref{def-01-K-tau-truncated} is essentially the form factor of a `collectively tuned' subsequence $\{ \varepsilon_{1}^\prime, \dots, \varepsilon_{L}^\prime\}$, where
  $\varepsilon_j^\prime = \varepsilon_{M+j}-\varepsilon_{M}$.
\hfill $\blacksquare$
\end{remark}

\begin{lemma}\label{PS-stationary-K-truncated}
  The spectral form factor of an ordered subsequence $\{ \varepsilon_{M+1} \le \dots \le \varepsilon_{M+L}\}$ of the complete spectrum $\{ 0 \le \varepsilon_1 \le \dots \le \varepsilon_{N}\}$ with \textbf{\textit{infinitly-stationary}} spacings equals
  \begin{eqnarray}\label{ps-def-K-tau-theorem} \fl
  K_{N; M, L}(\tau) = \frac{1}{L}
    \Bigg( \Big<
        \sum_{\ell=1}^{L} \sum_{m=1}^{L} e^{2 i \pi \tau (\varepsilon_\ell - \varepsilon_m)/\Delta }
    \Big>
    -  \Big<
        \sum_{\ell=1}^{L} e^{2 i \pi \tau \varepsilon_\ell/\Delta }
    \Big>
    \Big<
        \sum_{m=1}^L e^{-2 i \pi \tau \varepsilon_m/\Delta }
    \Big> \Bigg). \nonumber\\
    {}
\end{eqnarray}
\end{lemma}

\begin{proof}
    Since the {\it infinite stationarity} of level spacings, Definition~\ref{def-inf-stationary}, implies equality of moments $\langle (\varepsilon_\ell - \varepsilon_m)^q \rangle$ and $\langle \varepsilon_{\ell-m}^q \rangle$ for all $q \in {\mathbb N}$, one is allowed to simultaneously shift eigenvalues' indices in the differences $(\varepsilon_\ell - \varepsilon_m)$ appearing under the sign of averages $\langle\dots\rangle$ in Eq.~(\ref{ps-def-K-tau-truncated}) by the same integer $M$: $(\varepsilon_\ell
     -\varepsilon_m) \mapsto (\varepsilon_{\ell-M}
     -\varepsilon_{m-M})$ with $\varepsilon_0=0$. This observation produces Eq.~(\ref{ps-def-K-tau-theorem}).
\end{proof}

\begin{remark}
    As the form factor Eq.~(\ref{ps-def-K-tau-truncated}) of a subsequence $\{ \varepsilon_{M+1} \le \varepsilon_{M+2} \le  \dots \le  \varepsilon_{M+L}\}$ drawn from the {\it infinitely stationary} spectrum does {\it not} depend on the number $M$ of skipped low-lying eigenvalues, see Eq.~(\ref{ps-def-K-tau-theorem}), from now on, we shall use the notation $K_{N;L}(\tau)$ instead of $K_{N;M,L}(\tau)$.
\hfill $\blacksquare$
\end{remark}

\fancyhead{} \fancyhead[RE,LO]{\textsl{\thesection \quad Random diagonal matrices: Truncated eigenlevel sequences}}
\fancyhead[LE,RO]{\thepage}
\section{Random diagonal matrices: Spectral statistics of truncated eigenlevel sequences} \label{rdm-section}
\noindent
Random diagonal matrices (RDM) are routinely used to model spectral fluctuations in quantum systems exhibiting regular classical dynamics. Yet, much to our astonishment, we failed to find a comprehensive and transparent treatment of this paradigmatic model in the literature. In Section~\ref{RDM-unf} we partially bridge this gap: the issues of unfolding, ordering, infinite stationarity and correlations of individual eigenlevels are of our primary interest. These results, combined with a general framework of Section~\ref{PSSF}, are utilized in Sections~\ref{PS-RDM} and \ref{SFF-sec}, where the power spectrum and the spectral form factor are calculated nonperturbatively for both complete and truncated spectra of the finite--$N$ RDM-model. The two statistics are further analyzed in various scaling limits to lay the ground for Section~\ref{FFA-status}, where we discuss a status of the form factor approximation in the RDM context.

\subsection{Unfolding, ordering and (infinite) stationarity of level spacings} \label{RDM-unf}
\noindent
{\it Unfolding.}---Consider an ensemble of $N \times N$ random diagonal matrices of the form ${\bm H}={\rm diag}(\lambda_1,\dots,\lambda_N)$, where $\lambda_j \in {\mathbb R}$ are independent identically distributed random variables, each characterized by the probability density function $f_{\lambda_j}(\lambda)=f(\lambda)$ for all $j =1,\dots,N$. The mean spectral density of ${\bm H}$
\begin{eqnarray}
    \varrho_N (\lambda) = \langle \sum_{j=1}^N \delta(\lambda-\lambda_j) \rangle = N f(\lambda)
\end{eqnarray}
explicitly depends on the fluctuation properties of individual eigenlevels and, thus, is system-specific. Here, the angular brackets denote an appropriate ensemble averaging.

To study the {\it universal} properties of random diagonal matrices, one needs to re-scale their fluctuating eigenlevels in such a way that the average spacing between two adjacent levels becomes a constant. This can be achieved through the `unfolding procedure'~\cite{M-2004} which introduces a new set of eigenlevels according to the map
\begin{eqnarray} \label{map}
    \varepsilon_j = \int_{-\infty}^{\lambda_j} d\lambda\,\varrho_N(\lambda),
\end{eqnarray}
where $j = 1,\dots,N$ and $\varepsilon_j \in (0,N)$. The above definition ensures that all $\varepsilon_j$'s are independent, uniformly distributed random variables, $\varepsilon_j \sim {\rm U}(0,N)$. Indeed, the probability density function of $\varepsilon_j$ equals
\begin{eqnarray} \fl \qquad\qquad
    f_{\varepsilon_j}(\varepsilon) = \langle \delta(\varepsilon - \varepsilon_j)  \rangle = \Big< \delta\left(
         \int_{\lambda_j}^{\lambda} dE\,\varrho_N(E)
       \right) \Big> \nonumber\\
         \qquad \qquad =\frac{1}{\varrho_N(\lambda)} \langle \delta(\lambda-\lambda_j) \rangle
       = \frac{1}{N}\, \mathds{1}_{0\le \varepsilon \le N}.
\end{eqnarray}
The latter will now be denoted
\begin{eqnarray}
    f_{\varepsilon_j}(\varepsilon) = {\mathcal U}_N(\varepsilon) = \frac{1}{N}\, \mathds{1}_{0\le \varepsilon \le N}.
\end{eqnarray}
Here, $\mathds{1}_{X}$ is the indicator function; it equals unity if $X$ is true; otherwise it brings zero.
\noindent\newline\newline
{\it Ordering induces correlations.}---Aimed at studying statistical properties of truncated eigenlevel sequences introduced in Section \ref{PSSF}, we
further concentrate on {\it ordered} unfolded eigenlevels~\footnote[3]{Ordering and unfolding are mutually commuting procedures.}.

\begin{proposition}\label{prop-JPDF-ordered}
  Let $0\le \varepsilon_1\le \dots \le \varepsilon_N \le N$ be a set of unfolded eigenlevels generated by the map Eq.~(\ref{map}) out of ordered eigenvalues $\{ \lambda_1 \le \dots \le \lambda_N\}$ of a random diagonal matrix. Then, \newline\newline
  (i) the probability density $g_\ell(\varepsilon)$ of the $\ell$-th ordered eigenlevel equals
\begin{eqnarray} \label{gL-exp}
        g_\ell(\varepsilon) = \frac{1}{N} \frac{\Gamma(N+1)}{\Gamma(\ell) \Gamma(N+1-\ell)} \left( \frac{\varepsilon}{N} \right)^{\ell-1}
         \left( 1- \frac{\varepsilon}{N}\right)^{N-\ell} \mathds{1}_{0\le \varepsilon \le N},
\end{eqnarray}
that is
$$
    \frac{\varepsilon_\ell}{N} \sim {\rm Beta}_1(\ell, N+1 -\ell)
$$
is a univariate beta-variable;\newline\newline
(ii) the joint probability density $g_{\ell m}(\varepsilon,\varepsilon^\prime)$ of the $\ell$-th and $m$-th ordered eigenlevels ($\ell< m$) equals
\begin{eqnarray}\fl \label{gLM-exp}
    \qquad g_{\ell m} (\varepsilon,\varepsilon^\prime) = \frac{1}{N^2} \frac{\Gamma(N+1)}{\Gamma(\ell) \Gamma(m-\ell) \Gamma(N+1-m)} \nonumber\\
    \qquad \times
    \left( \frac{\varepsilon}{N}\right)^{\ell-1} \left( \frac{\varepsilon^\prime}{N} - \frac{\varepsilon}{N}\right)^{m-\ell-1}
    \left( 1 -  \frac{\varepsilon^\prime}{N} \right)^{N-m} \mathds{1}_{0\le \varepsilon \le \varepsilon^\prime \le N},
\end{eqnarray}
that is
$$
    \left( \frac{\varepsilon_\ell}{N}, \frac{\varepsilon_m}{N} \right) \sim {\rm Beta}_2(\ell, m-\ell, N+1 -m)
$$
is a bivariate beta-variable~\cite{M-1959}.
\end{proposition}
\begin{proof}
  Given the joint probability density of all $N$ unfolded (not ordered!) eigenvalues $\prod_{j=1}^N {\mathcal U}_N(\varepsilon_j)$, we observe
  \begin{eqnarray} \label{g1-interm} \fl
    g_\ell(\varepsilon) = {{N}\choose{\ell-1,1,N-\ell}} \left(\prod_{j=1}^{\ell-1} \int_{0}^{\varepsilon} d\varepsilon_j
    {\mathcal U}_N(\varepsilon_j)\right) {\mathcal U}_N(\varepsilon) \left(\prod_{k=\ell+1}^{N} \int_{\varepsilon}^{N} d\varepsilon_k
    {\mathcal U}_N(\varepsilon_k)\right)
  \end{eqnarray}
and
  \begin{eqnarray} \label{g2-interm} \fl
    g_{\ell m}(\varepsilon, \varepsilon^\prime) = {{N}\choose{\ell-1,1,m-\ell-1, 1, N-m}} \left(\prod_{j=1}^{\ell-1} \int_{0}^{\varepsilon} d\varepsilon_j
    {\mathcal U}_N(\varepsilon_j)\right) {\mathcal U}_N(\varepsilon) \nonumber\\
     \times  \left(\prod_{k=\ell+1}^{m-1} \int_{\varepsilon}^{\varepsilon^{\prime}} d\varepsilon_k
    {\mathcal U}_N(\varepsilon_k)\right) {\mathcal U}_N(\varepsilon^\prime)
    \left(\prod_{n=m+1}^{N} \int_{\varepsilon^\prime}^{N} d\varepsilon_n
    {\mathcal U}_N(\varepsilon_n)\right),
  \end{eqnarray}
where
$$
  {{N}\choose{\ell_1,\dots,\ell_n}} = \frac{\Gamma(N+1)}{\prod_{j=1}^n \Gamma(\ell_j+1)}
$$
is a multinomial coefficient. Performing simple integrals in Eqs.~(\ref{g1-interm}) and (\ref{g2-interm}), we reproduce Eqs.~(\ref{gL-exp}) and (\ref{gLM-exp}), respectively.
\end{proof}
\noindent\newline
{\it Infinite stationarity of level spacings.}---Interestingly, a sequence of level spacings associated with unfolded and ordered eigenlevels in an ensemble of
random diagonal matrices is {\it infinitely stationary}, see Definition~\ref{def-inf-stationary}.

\begin{proposition}\label{prop-inf}
  Let $0\le \varepsilon_1\le \dots \le \varepsilon_N \le N$ be a set of unfolded eigenlevels generated by the map Eq.~(\ref{map}) out of ordered eigenvalues $\{ \lambda_1 \le \dots \le \lambda_N\}$ of a random diagonal matrix. Then an associated sequence of level spacings $\{s_1,\dots,s_N\}$ between consecutive eigenlevels
  ($s_\ell=\varepsilon_\ell - \varepsilon_{\ell-1}$ with $\ell = 1,\dots, N$ and $\varepsilon_0=0$) is infinitely stationary in the sense of Definition~\ref{def-inf-stationary}.
\end{proposition}
\begin{proof}
  Use Eqs.~(\ref{gL-exp}) and (\ref{gLM-exp}) of Proposition~\ref{prop-JPDF-ordered} to evaluate ($q \ge 0$)
\begin{eqnarray} \label{1-m}
    \langle \varepsilon_k^q\rangle = \int_{0}^{N} d\varepsilon \, \varepsilon^q g_k(\varepsilon) = N^q \frac{\Gamma(N+1)}{\Gamma(N+1+q)}\frac{\Gamma(k+q)}{\Gamma(k)},
\end{eqnarray}
and ($\ell > m$)
\begin{eqnarray}\label{2-m}
    \langle (\varepsilon_\ell -\varepsilon_m)^q\rangle &=& \int_{0}^{N} d\varepsilon \int_{0}^{N} d\varepsilon^\prime \, (\varepsilon-\varepsilon^\prime)^q
    g_{m \ell}(\varepsilon^\prime,\varepsilon)\nonumber\\
     &=& N^q \frac{\Gamma(N+1)}{\Gamma(N+1+q)}\frac{\Gamma(\ell-m+q)}{\Gamma(\ell-m)}.
\end{eqnarray}
Comparing Eq.~(\ref{1-m}) taken at $k=\ell-m$ with Eq.~(\ref{2-m}), we observe the equality as required by Eq.~(\ref{L-1-inf}) of Definition~\ref{def-inf-stationary}.
\end{proof}

\subsection{Power spectrum of complete and truncated sequences}\label{PS-RDM}
Now that the spadework has all been done, an exact calculation of the power spectrum becomes starightforward.
\noindent\newline\newline
{\it Complete spectrum.}---Due to stationarity of level spacings proven in Proposition~\ref{prop-inf}, the power spectrum of a complete eigenvalue sequence can be determined by virtue of Theorem~\ref{PS-stationary-main}. Since both the mean
\begin{eqnarray} \label{m-L}
    \langle \varepsilon_{\ell} \rangle = \frac{N}{N+1}\,\ell
\end{eqnarray}
and the covariance
\begin{eqnarray}\label{v-L}
    {\rm cov}[\varepsilon_\ell,\varepsilon_m] =  \frac{N^2}{(N+1)(N+2)}\left(
        {\rm min}(\ell,m) - \frac{\ell m}{N+1}
    \right)
\end{eqnarray}
can be read off from Eqs.~(\ref{1-m}) and (\ref{2-m}), the power spectrum of a complete sequence can be calculated from Eq.~(\ref{smd-sum}) after substituting
\begin{eqnarray} \label{var-over-delta2}
    \frac{1}{\Delta^2}{\rm var}[\varepsilon_\ell] =  \frac{N+1}{N+2}\ell \left(
        1 - \frac{\ell}{N+1}
    \right)
\end{eqnarray}
therein.

\begin{remark}
The first, linear in $\ell$, term in Eq.~(\ref{var-over-delta2}) essentially coincides with the variance of $\ell$-th ordered eigenvalue in a simple model of spectral sequences with identically distributed {\it uncorrelated} spacings considered in Section 1.3 of Ref.~\cite{ROK-2020}. The second, {\it negative} (and quadratic in $\ell$), term in Eq.~(\ref{var-over-delta2}) is a fingerprint of weak albeit long-range {\it negative} correlations between the {\it spacings} which appear to be an intrinsic feature~\footnote[8]{Indeed, Eqs.~(\ref{m-L}) and (\ref{v-L}) imply that the Pearson correlation coefficient $\rho[s_\ell,s_m]$ of $\ell$-th and $m$-th level {\it spacings} equals
$$
    \rho[s_\ell,s_m] = \left\{
                         \begin{array}{cc}
                           \displaystyle 1, & \hbox{$\ell=m$;} \\
                           \displaystyle -1/N, & \hbox{$\ell \neq m$.}
                         \end{array}
                       \right.
$$} of spectral fluctuations in random diagonal matrices. Hence, as soon as the {\it complete} spectra are concerned, the power spectrum in a model of random diagonal matrices should {\it differ} from the one in a model of spectral sequences with uncorrelated spacings. Explicit calculations confirm this expectation.
\hfill $\blacksquare$
\end{remark}

\begin{proposition}\label{RDM-PS-complete}
   Let $\{0 \le \varepsilon_1 \le \dots \le \varepsilon_N \le N\}$ be a complete sequence of unfolded eigenlevels, $N \in {\mathbb N}$, generated by the map
   Eq.~(\ref{map}) out of ordered eigenvalues $\{ \lambda_1 \le \dots \le \lambda_N\}$ of a random diagonal matrix. The power spectrum of such a sequence equals
\begin{eqnarray}\label{smd-sum-33}
    S_N(\omega) = \frac{N+1}{N+2}\left( S_N^{\rm{(u)}}(\omega) +  S_N^{\rm{(c)}}(\omega) \right),
\end{eqnarray}
where
\begin{eqnarray} \label{SNU}
    S_N^{\rm{(u)}}(\omega) =  \frac{2N+1}{4N} \frac{1}{\sin^2(\omega/2)} \left(
        1 - \frac{1}{2N+1} \frac{\sin\left((N+1/2)\omega\right)}{\sin(\omega/2)}
    \right)
\end{eqnarray}
and
\begin{eqnarray} \label{SNC}
  S_N^{\rm{(c)}}(\omega)  = - \frac{1}{4 \sin^2(\omega/2)} \frac{N}{N+1} \nonumber\\
  \quad\quad \times \left\{
        1 - \frac{2}{N} \frac{\sin(\omega N/2)}{\sin(\omega/2)} \cos\left( \frac{\omega}{2}(N+1)\right) + \frac{1}{N^2}\frac{\sin^2(\omega N/2)}{\sin^2(\omega/2)}
    \right\}.
\end{eqnarray}
\end{proposition}
\begin{proof}
Make use of Theorem~\ref{PS-stationary-main}, Proposition~\ref{prop-inf} and substitute Eq.~(\ref{var-over-delta2}) into Eq.~(\ref{smd-sum}) to complete the proof.
\end{proof}
\begin{remark} \label{remark-origin-ps}
  In Eq.~(\ref{smd-sum-33}), the contribution proportional to $S_N^{\rm{(u)}}(\omega)$ corresponds to the power spectrum of $N$ consecutive eigenlevels as if their spacings, distributed identically, were uncorrelated. The second -- {\it negative} -- contribution determined by the function $S_N^{\rm{(c)}}(\omega)$, originates from weak negative long-range correlations of level spacings characteristic of ordered spectra of random diagonal matrices.
\hfill $\blacksquare$
\end{remark}
\begin{remark}
As $N\rightarrow \infty$, two scaling limits of the power spectrum $S_N(\omega)$ can readily be identified in two regimes: in the infrared frequency domain $\omega={\mathcal O}(N^{-1})$ and in the bulk $\omega={\mathcal O}(N^0)$. Corresponding limiting laws can be read off from Corollary~\ref{PS-SL-trunc} below upon setting $\alpha_L=1$ therein.
\hfill $\blacksquare$
\end{remark}
\noindent\newline
{\it Truncated spectrum.}---Following the discussion in Section~\ref{PS-trunc-ss}, we adopt Definition~\ref{def-01-truncated} of the power spectrum for truncated subsequence $\{\varepsilon_{M+1} \le \dots \le \varepsilon_{M+L}\}$ of the complete spectrum $\{0 \le \varepsilon_1 \le \dots \le \varepsilon_N \}$. The power spectrum of such a truncated subsequence is determined as follows.

\begin{theorem}\label{SNL-theorem}
   Let $\{0 \le \varepsilon_1 \le \dots \le \varepsilon_N \le N\}$ be a complete sequence of unfolded eigenlevels, $N \in {\mathbb N}$, generated by the map
   Eq.~(\ref{map}) out of ordered eigenvalues $\{ \lambda_1 \le \dots \le \lambda_N\}$ of a random diagonal matrix, and let $\{ \varepsilon_{M+1} \le \dots \le \varepsilon_{M+L}\}$ be a spectral subsequence of the length $L \in {\mathbb N}$ obtained from the complete sequence by omitting both $M\in {\mathbb N} \cup \{0\}$ low-lying eigenlevels $\{\varepsilon_1 \le \dots \le \varepsilon_{M}\}$ and $N-M-L$ high-lying eigenlevels $\{ \varepsilon_{M+L+1} \le \dots \le \varepsilon_{N}\}$, such that $M+L\le N$. The power spectrum of the subsequence $\{ \varepsilon_{M+1} \le \dots \le \varepsilon_{M+L}\}$ equals
\begin{eqnarray}\label{smd-sum-L}
    S_{N;L}(\omega) = \frac{N+1}{N+2} S_L^{\rm{(u)}}(\omega) + \frac{L+1}{N+2} S_L^{\rm{(c)}}(\omega),
\end{eqnarray}
where
\begin{eqnarray} \label{SLu}
    S_L^{\rm{(u)}}(\omega) =  \frac{2L+1}{4L} \frac{1}{\sin^2(\omega/2)} \left(
        1 - \frac{1}{2L+1} \frac{\sin\left((L+1/2)\omega\right)}{\sin(\omega/2)}
    \right)
\end{eqnarray}
and
\begin{eqnarray} \label{SLc}
  S_L^{\rm{(c)}}(\omega)  = - \frac{1}{4 \sin^2(\omega/2)} \frac{L}{L+1}\nonumber\\
  \qquad \times \left\{
        1 - \frac{2}{L} \frac{\sin(\omega L/2)}{\sin(\omega/2)} \cos\left( \frac{\omega}{2}(L+1)\right) + \frac{1}{L^2}\frac{\sin^2(\omega L/2)}{\sin^2(\omega/2)}
    \right\}.
\end{eqnarray}
\end{theorem}
\noindent
\begin{proof}
Make use of Theorem~\ref{PS-stationary-truncated}, and substitute Eq.~(\ref{var-over-delta2}) into Eq.~(\ref{smd-sum-x}) to complete the proof.
\end{proof}
\begin{remark}
Equations (\ref{smd-sum-L}), (\ref{SLu}) and (\ref{SLc}) provide an exact solution for the power spectrum of truncated eigenvalue sequences in random diagonal matrices. Notice that this result holds after the entire sub-sequence $\{\varepsilon_{M+1} \le \dots \le \varepsilon_{M+L}\}$ has been collectively tuned by subtracting a `reference level' $\varepsilon_M$ from each eigenlevel of subsequence, see Remark~\ref{rem-collective}. In absence of such a tuning, the power spectrum will be distorted by parasitic oscillations originating from collective fluctuations of the subsequence as a whole, as discussed in Section~\ref{PS-trunc-ss}. An explicit formula for oscillatory signal $\delta\tilde{S}_{N;M,L}(\omega)$ can be determined from Eq.~(\ref{d-SML-new}). Indeed, substituting Eq.~(\ref{var-over-delta2}) therein, we derive:
\begin{eqnarray} \label{dS-osc}
    \delta\tilde{S}_{N;M,L}(\omega) =  \frac{M}{L}
          \frac{N-M-L}{N+2} \frac{\sin^2(\omega L/2)}{\sin^2(\omega/2)}.
\end{eqnarray}
Its amplitude can be significantly higher than the power spectrum $S_{N;L}(\omega)$ itself; see also Eq.~(\ref{SML-new-2}) and a discussion around it.
\hfill $\blacksquare$
\end{remark}
\begin{remark} \label{rem-correlations}
For the origin of the two terms in the power spectrum Eq.~(\ref{smd-sum-L}), see Remark~\ref{remark-origin-ps}. This time, however, a negative contribution $S_L^{\rm{(c)}}(\omega)$ deriving from weak long-range correlations between level spacings, bears a weight factor $(L+1)/(N+2)$. Since it becomes vanishingly small for subsequences of finite length $L$ as $N \rightarrow \infty$, the power spectrum of RDM sequences of finite length reduces to the power spectrum of finite eigenvalue sequences with uncorrelated spacings, see Section 1.3 of Ref.~\cite{ROK-2020}.
\hfill $\blacksquare$
\end{remark}

\begin{corollary}[Finite subsequence]\label{PS-inf-L-omega}
  Let $S_{N;L}(\omega)$ be the power spectrum of a spectral subsequence of length $L$ obtained from a complete spectrum of random diagonal matrices as specified in Theorem~\ref{SNL-theorem}. For $\omega={\mathcal O}(N^0)$ and $L={\mathcal O}(N^0)$, it holds that
\begin{eqnarray} \label{SL-omega}\fl \quad
    S_{\infty;L}(\omega) = \lim_{N\rightarrow \infty} S_{N;L}(\omega) \nonumber\\
    \fl\qquad\qquad\quad =
    \frac{2L+1}{4L} \frac{1}{\sin^2(\omega/2)} \left(
        1 - \frac{1}{2L+1} \frac{\sin\left((L+1/2)\omega\right)}{\sin(\omega/2)}
    \right).
\end{eqnarray}
\end{corollary}
\begin{remark}\label{R-S}
Equation~(\ref{SL-omega}) coincides with the power spectrum calculated in Ref.~\cite{ROK-2020} for random spectra with uncorrelated level spacings (with unit mean and variance).
As $L\rightarrow \infty$, two scaling limits of the power spectrum $S_{\infty;L}(\omega)$ have been identified in Ref.~\cite{ROK-2020} -- in the infrared frequency domain $\omega={\mathcal O}(L^{-1})$ and in the bulk $\omega={\mathcal O}(L^0)$. Corresponding limiting laws can alternatively be recovered from Eqs.~(\ref{S-minus-1}) and (\ref{S-0}) of Corollary~\ref{PS-SL-trunc} upon {\it formal} setting $\alpha_L=0$ therein.
\hfill $\blacksquare$
\end{remark}

In Fig.~\ref{Plot_2}, our theoretical predictions for regularized~\footnote[2]{~To get rid of the singularity of $S_{N;L(N)}(\omega)$ at $\omega=0$, we consider the combination $\omega^2 S_{N;L(N)}(\omega)$ instead of the power spectrum itself.} power spectrum are confronted with the results of numerical simulations performed for complete and truncated spectra of random diagonal matrices. Referring the reader to a figure caption for detailed explanations, we plainly notice a perfect agreement between the simulations and the theoretical results.

\begin{figure}
\includegraphics[width=\textwidth]{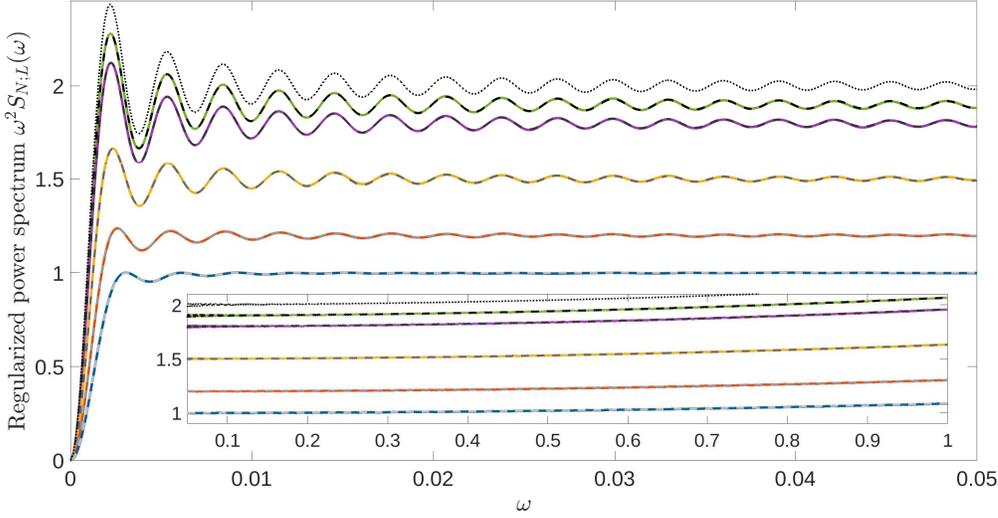}
\caption{
Regularized power spectrum $\omega^2 S_{N;L}(\omega)$ as a function of frequency $\omega$ for truncated RDM eigenlevel sequences characterized by various truncation parameters $\alpha_L$. (i)~Solid curves represent regularized power spectra computed numerically by sampling $L=2048$ consecutive eigenlevels out of complete spectra of random diagonal matrices whose $N=L/\alpha_L$ independent eigenvalues were uniformly distributed ${\rm U}(0, N)$. Averaged over $M=10^6$ realizations, the green, purple, yellow, red and blue curves mark the regularized power spectrum for truncation parameters $\alpha_L = 0.1$, $0.2$, $0.5$, $0.8$, and $1$, respectively. (ii)~Dashed lines on the top of colored ones display theoretical curves from Theorem~\ref{SNL-theorem}. (iii)~Dotted line represents a theoretical curve for the regularized power spectrum for eigenlevel sequences of the length $L=2048$ with independent and exponentially distributed level spacings (Corollary~\ref{PS-inf-L-omega} and Remark~\ref{R-S}). Its comparison with a simulated curve can be found in Fig.~1 of Ref.~\cite{ROK-2020}. Inset: the same graphs for higher frequencies.
}
\label{Plot_2}
\end{figure}

Motivated by a forthcoming discussion of validity of the form factor approximation in Section~\ref{FFA-status}, below we are going to concentrate on two scaling limits of the power spectrum for truncated eigenlevel sequences of random diagonal matrices as their dimensionality $N\rightarrow\infty$.
\begin{corollary} \label{PS-SL-trunc}
  Let $M=M(N)$ be the number of skipped low-lying eigenlevels and $L=L(N)$ be the length of the spectral subsequence such that $M(N)+L(N) \le N$. If $L(N) \rightarrow \infty$ as $N\rightarrow \infty$ and the two limits
\begin{eqnarray}\label{scales-LMN}
    \alpha_L = \lim_{N\rightarrow\infty} \frac{L(N)}{N}, \quad \alpha_M = \lim_{N\rightarrow\infty} \frac{M(N)}{N}
\end{eqnarray}
exist, then the regularized power spectrum $\omega^2 S_{N;L(N)}(\omega)$ admits the following scaling limits:
\begin{eqnarray}\label{S-minus-1}
    {\mathcal S}_{\alpha_L}^{(-1)}(\Omega) &=& \lim_{N\rightarrow\infty} \omega^2 S_{N;L(N)}(\omega)\Big|_{\omega=\Omega/L(N)} \nonumber\\
    &=&  2-\alpha_L -\alpha_L \frac{\sin^2(\Omega/2)}{(\Omega/2)^2} -2 (1-\alpha_L) \frac{\sin\Omega}{\Omega},
\end{eqnarray}
and
\begin{eqnarray}\label{S-0}
    {\mathcal S}_{\alpha_L}^{(0)}(\omega) &=& \lim_{N\rightarrow\infty}\omega^2 S_{N;L(N)}(\omega) \nonumber\\
        &=&
     \frac{2-\alpha_L}{4} \frac{\omega^2}{\sin^2(\omega/2)}, \quad \omega_0 < \omega \le \pi,
\end{eqnarray}
where $\omega_0>0$ can be arbitrary small but fixed. Furthermore, any other scaling limit of $S_{N;L(N)}$ is trivial, i.e. for $\omega=\tilde{\Omega}/L(N)^\delta$ with the scaling exponent $0<\delta<1$, we have:
\begin{eqnarray}\label{S-minus-lambda}
    {\mathcal S}_{\alpha_L}^{(-\delta)}(\tilde{\Omega}) = \lim_{N\rightarrow\infty} \omega^2 S_{N;L(N)}(\omega)\Big|_{\omega=\tilde{\Omega}/L(N)^\delta}
    = 2-\alpha_L.
\end{eqnarray}
\end{corollary}
\noindent\newline
{\it Discussion.}---Corollary~\ref{PS-SL-trunc} indicates that the power spectrum of random diagonal matrices is affected by eigenlevel truncations at ${\it all}$ frequencies: the power spectrum depends on the truncation parameter $\alpha_L$ in both the infrared region $\omega={\mathcal O}(L(N)^{-1})$ and the bulk $\omega={\mathcal O}(L(N)^0)$, see Eqs.~(\ref{S-minus-1}) and (\ref{S-0}), respectively. This observation has important implications for a quantitative description of the power spectrum in closed quantum systems with integrable classical dynamics.

As far as a {\it complete spectrum} of random diagonal matrices is concerned, the limiting laws for the power spectrum are furnished by Corollary~\ref{PS-SL-trunc} taken at  $\alpha_L=1$. In this case, the power spectrum in the bulk and near the origin is described by
\begin{eqnarray} \label{S-int-qs-bulk-N}
    S_{N}(\omega) = \frac{1}{4} \times \frac{1}{\sin^2(\omega/2)} + {\mathcal O}(N^{-2}), \quad \omega_0 < \omega \le \pi,
\end{eqnarray}
and
\begin{eqnarray} \label{S-int-qs-infrared-N}
    \omega^2 S_{N}(\omega)\Big|_{\omega=\Omega/N} = 1 - \frac{\sin^2(\Omega/2)}{(\Omega/2)^2} + {\mathcal O}(N^{-1}),
\end{eqnarray}
respectively, where $\omega_0>0$ is arbitrarily small but fixed. The remainder terms in Eqs.~(\ref{S-int-qs-bulk-N}) and ~(\ref{S-int-qs-infrared-N}) follow from Proposition~\ref{RDM-PS-complete}.

Even though Eqs.~(\ref{S-int-qs-bulk-N}) and (\ref{S-int-qs-infrared-N}) have been derived for random diagonal matrices (which are widely believed to mimic spectral statistics in quantum systems with integrable classical dynamics), they do {\it not} describe the power spectrum in realistic quantum systems. Indeed, in any conceivable experimental setup, only a {\it finite} number $L={\mathcal O}(N^0)$ of eigenlevels -- out of {\it infinitely many} ($N\rightarrow\infty$) --  can be measured. This situation is described by Corollary~\ref{PS-inf-L-omega} which brings, as $L\rightarrow \infty$, {\it different} power spectrum laws in the bulk
\begin{eqnarray} \label{S-int-qs-bulk}
    S_{\infty;L}(\omega) = \frac{1}{2} \times \frac{1}{\sin^2(\omega/2)} + {\mathcal O}(L^{-1}), \quad \omega_0 < \omega \le \pi,
\end{eqnarray}
and near the origin $\omega=0$:
\begin{eqnarray} \label{S-int-qs-infrared}
    \omega^2 S_{\infty;L}(\omega)\Big|_{\omega=\Omega/L} = 2 \left(1 - \frac{\sin\Omega}{\Omega}\right) + {\mathcal O}(L^{-1}).
\end{eqnarray}

Importantly, {\it both} limiting laws [Eq.~(\ref{S-int-qs-bulk}) and (\ref{S-int-qs-infrared})] differ from their counterparts [Eqs.~(\ref{S-int-qs-bulk-N}) and (\ref{S-int-qs-infrared-N})] derived for complete RDM spectra. In particular, comparison of the leading terms in Eqs.~(\ref{S-int-qs-bulk-N}) and (\ref{S-int-qs-bulk}) reveals that, {\it at any finite frequency} (that is, away from the origin $\omega=0$), the power spectrum of infinite-dimensional random diagonal matrices is {\it twice smaller} than the power spectrum of a spectral subsequence of a large~\footnote[5]{This holds true in a wider setting, when a subsequence length $L(N)$ grows weaker than $N$ such that $L(N)/N\rightarrow 0$ as $N\rightarrow \infty$, see Corollary~\ref{PS-SL-trunc}.}  but {\it finite} length. The origin of such a reduction was highlighted in Remark~\ref{rem-correlations}.

This discussion leads us to conclude that the bulk of the power spectrum in generic quantum systems with completely regular classical geodesics is described by Eq.~(\ref{S-int-qs-bulk}) (notice the factor $\rfrac{1}{2}$) rather than by Eq.~(\ref{S-int-qs-bulk-N}) (notice the factor $\rfrac{1}{4}$) previously claimed in Refs.~\cite{FGMMRR-2004}. This conclusion is in concert with our previous findings~\cite{ROK-2017,ROK-2020}. Numerical studies of the power spectrum in rectangular and semicircular billiards, to be presented in Section~\ref{Num-Billiards}, lend independent support to the universal laws Eqs.~(\ref{S-int-qs-bulk}) and (\ref{S-int-qs-infrared}).

\subsection{Spectral form factor of complete and truncated sequences}\label{SFF-sec}

Motivated by recent discussions~\cite{ROK-2017,ROK-2020} of the status of a form factor approximation in the theories of the power spectrum of bounded quantum systems, we now turn to a nonperturbative analysis of the form factor in the model of random diagonal matrices; both complete and truncated eigenspectra will be considered.
\noindent\newline\newline
{\it Complete spectrum.}---To determine the form factor for a complete spectrum of random diagonal matrices, we adopt Definition~\ref{def-01-K-tau-truncated}:
\begin{eqnarray} \fl \label{ps-K-tau}
    K_N(\tau)=K_{N;N}(\tau) \nonumber\\
    \fl\qquad =
    \frac{1}{N}
    \left( \Big<
        \sum_{\ell=1}^N \sum_{m=1}^N e^{2 i \pi \tau (\varepsilon_\ell - \varepsilon_m)/\Delta }
    \Big> -  \Big<
        \sum_{\ell=1}^N e^{2 i \pi \tau \varepsilon_\ell/\Delta }
    \Big>
    \Big<
        \sum_{m=1}^N e^{-2 i \pi \tau \varepsilon_m/\Delta }
    \Big> \right).
\end{eqnarray}
Here, summations run over a set $\{\varepsilon_1,\dots,\varepsilon_N\}$ of either ordered or unordered unfolded eigenlevels.

\begin{proposition}\label{RDM-FF-complete}
   Let $\{\varepsilon_1, \dots, \varepsilon_N\}$ be a complete sequence of unfolded eigenlevels, $N \in {\mathbb N}$, generated by the map
   Eq.~(\ref{map}) out of the eigenvalues $\{ \lambda_1, \dots, \lambda_N\}$ of a random diagonal matrix. The spectral form factor of such a sequence equals
\begin{eqnarray} \label{FF}
    \quad
    K_N(\tau) = 1 - \left(\frac{\sin[\pi \tau (N+1)]}{\pi \tau (N+1)}\right)^2.
\end{eqnarray}
\end{proposition}
\begin{proof}
Considering the sequence $\{\varepsilon_1, \dots, \varepsilon_N\}$ in absence of ordering, we follow the discussion of Section~\ref{RDM-unf} to conclude that they are uniform i.i.d. random variables, $\varepsilon_\ell \sim {\rm U}(0,N)$ with the mean level spacing $\Delta = N/(N+1)$, see Eq.~(\ref{m-L}). Observing that
\begin{eqnarray}
    \langle e^{2 i \pi \tau \varepsilon_\ell/\Delta } \rangle = e^{i\pi \tau (N+1)} \frac{\sin[\pi \tau (N+1)]}{\pi \tau (N+1)},
\end{eqnarray}
we substitute it into Eq.~(\ref{ps-K-tau}) to derive Eq.~(\ref{FF}).
\end{proof}
\noindent\newline
{\it Truncated spectrum.}---Infinite stationarity of level spacings in ensemble of random diagonal matrices makes it possible to determine the spectral form factor exactly also for truncated sequences. This being said, a calculation of the form factor for truncated spectrum is way more involved since eigenlevel ordering becomes essential.

\begin{theorem}\label{KNL-theorem}
   Let $\{0 \le \varepsilon_1 \le \dots \le \varepsilon_N \le N\}$ be a complete sequence of unfolded eigenlevels, $N \in {\mathbb N}$, generated by the map
   Eq.~(\ref{map}) out of ordered eigenvalues $\{ \lambda_1 \le \dots \le \lambda_N\}$ of a random diagonal matrix, and let $\{ \varepsilon_{M+1} \le \dots \le \varepsilon_{M+L}\}$ be a spectral subsequence of the length $L \in {\mathbb N}$ obtained from the complete sequence by omitting both $M\in {\mathbb N} \cup \{0\}$ low-lying eigenlevels $\{\varepsilon_1 \le \dots \le \varepsilon_{M}\}$ and $N-M-L$ high-lying eigenlevels $\{ \varepsilon_{M+L+1} \le \dots \le \varepsilon_{N}\}$, such that $M+L\le N$. The form factor of the subsequence $\{ \varepsilon_{M+1} \le \dots \le \varepsilon_{M+L}\}$ equals
        \begin{eqnarray} \fl \quad \label{KLN}
        K_{N;L}(\tau) = \frac{1}{L} \Bigg\{
        L + \frac{N(N-1)}{2\pi^2 \tau^2 (N+1)^2} \Big(
        1 - {\rm Re} {}_1 F_{1}(L-1;N-1; 2i \pi \tau (N+1))
      \Big) \nonumber\\
      \quad
      -\frac{N^2}{4 \pi^2 \tau^2 (N+1)^2} \Big|
        1 - {}_1 F_{1}(L;N; 2i \pi \tau (N+1))
      \Big|^2
    \Bigg\},
\end{eqnarray}
where ${}_1 F_1 (a;b;z)$ is the Kummer confluent hypergeometric function.
\end{theorem}
\noindent
\begin{proof}
By virtue of Lemma~\ref{PS-stationary-K-truncated} and Proposition~\ref{prop-inf}, the form factor $K_{N;L}(\tau)$ of the truncated subsequence is defined by Eq.~(\ref{ps-def-K-tau-theorem}), where the averaging runs over an ensemble of random diagonal matrices with the mean level spacing $\Delta=N/(N+1)$, see Eq.~(\ref{m-L}). To facilitate the calculation of $K_{N;L}(\tau)$ we spot that it can be written in the form
\begin{eqnarray} \label{K-FG}
    K_{N;L}(\tau) =\frac{1}{L} \left\{
        {\mathcal G}_{N;L}\left(\tau\frac{N+1}{N}\right) - \left| {\mathcal F}_{N;L}\left(\tau\frac{N+1}{N}\right) \right|^2
    \right\},
\end{eqnarray}
where
\begin{eqnarray} \label{F-sum}
    {\mathcal F}_{N;L}(\tau) = \left<
    \sum_{\ell=1}^L e^{2 i\pi \tau \varepsilon_\ell}
    \right>
\end{eqnarray}
and
\begin{eqnarray} \label{G-sum}
    {\mathcal G}_{N;L}(\tau) = \left<
    \sum_{\ell=1}^L \sum_{m=1}^L e^{2 i\pi \tau (\varepsilon_\ell -\varepsilon_m)}
    \right> = L + 2 {\rm Re} \sum_{\ell=1}^{L-1} {\mathcal F}_{N;\ell} (\tau).
\end{eqnarray}
The latter equality is a direct consequence of the infinite stationarity of level spacings in ensemble of random diagonal matrices, see Proposition~\ref{prop-inf}.

Having expressed the form factor in term of a single function ${\mathcal F}_{N;L} (\tau)$, we now turn to its evaluation. To proceed, we notice that it can be expressed
\begin{eqnarray} \label{F-sum-fourier}
    {\mathcal F}_{N;L}(\tau) = \int_{0}^{N} d\varepsilon \, \varrho_{N;L}(\varepsilon) \, e^{2 i \pi \tau \varepsilon}
\end{eqnarray}
in terms of the partial mean spectral density of the first $L$ ordered eigenvalues
\begin{eqnarray} \label{rho_L_N_sum}
        \varrho_{N;L}(\varepsilon) = \langle \sum_{\ell=1}^L \delta(\varepsilon -\varepsilon_\ell) \rangle =
        \sum_{\ell=1}^L g_\ell (\varepsilon),
\end{eqnarray}
where $g_\ell (\varepsilon)$ is the mean density of the $\ell$-th ordered eigenlevel, see Eq.~(\ref{gL-exp}), so that
\begin{eqnarray} \label{RNL-1}
        \varrho_{N;L}(\varepsilon) =\sum_{j=0}^{L-1} \frac{(N-1)!}{j!(N-1-j)!} \left(\frac{\varepsilon}{N}\right)^j
        \left(
            1- \frac{\varepsilon}{N}
        \right)^{N-1-j}.
\end{eqnarray}
Hence, the partial mean spectral density coincides with the cumulative distribution function of the binomial random variable $X_{\{\varepsilon/N,N-1\}} \sim {\rm Bin}(\varepsilon/N, N-1)$:
\begin{eqnarray} \fl \label{rho_L_N}
    \qquad
    \varrho_{N;L}(\varepsilon) = {\rm Prob}\left(
        X_{\{\varepsilon/N,N-1\}} \le L-1
    \right) = \frac{B(1-\varepsilon/N; N-L,L)}{B(N-L,L)},
\end{eqnarray}
where
\begin{eqnarray} \label{RNL-2}
    B(x;a,b) = \int_{0}^{x} t^{a-1} (1-t)^{b-a} dt
\end{eqnarray}
is the incomplete beta function and
\begin{eqnarray}  \label{RNL-3}
    B(a,b) = B(1;a,b) = \frac{\Gamma(a)\Gamma(b)}{\Gamma(a+b)}.
\end{eqnarray}
Substitution of Eq.~(\ref{rho_L_N}) into Eq.~(\ref{F-sum-fourier}) followed by integration by parts yields
\begin{eqnarray} \label{FLN-parts}
\fl
\qquad
    {\mathcal F}_{N;L}(\tau) =  \frac{1}{2 i \pi \tau} \Bigg(
        \frac{1}{B(N-L,L)}\int_{0}^{1} e^{2 i \pi \tau N \xi} \xi^{L-1} (1-\xi)^{N-L-1} d\xi - 1
    \Bigg).
\end{eqnarray}
The remaining integral is nothing but a
characteristic function for the random variable $Y \sim {\rm Beta}_1(L,N-L)$. This observation brings
\begin{eqnarray} \label{beta-major}
    {\mathcal F}_{N;L}(\tau) = \frac{1}{2 i \pi \tau} \Big(
         {}_1 F_{1}(L;N; 2 i \pi \tau N) - 1
    \Big).
\end{eqnarray}
The other sought function, ${\mathcal G}_{N;L}(\tau)$, can be determined from Eq.~(\ref{G-sum}). Substituting Eq.~(\ref{beta-major}) therein and making use of
the summation formula given by Eq.~(5.3.5.1) of Ref. \cite{PBM-2002},
\begin{eqnarray} \label{sum-rule}
    \sum_{\ell=1}^{L} {}_1 F_{1}(\ell;N; z) = \frac{N-1}{z} \Big(
        -1 + {}_1 F_{1}(L;N-1; z)
    \Big),
\end{eqnarray}
we derive:
\begin{eqnarray} \label{GLN-final}
 {\mathcal G}_{N;L}(\tau) =  L + \frac{N-1}{2 N \pi^2 \tau^2} \Big(
    1 - {\rm Re\,} {}_1 F_{1}(L-1;N-1; 2 i \pi \tau N)
 \Big).
\end{eqnarray}
Equation~(\ref{KLN}) follows from Eqs.~(\ref{K-FG}), (\ref{beta-major}) and (\ref{GLN-final}).
\end{proof}
\begin{remark}
  For $L=N$ the identity ${}_1 F_{1}(N;N; z) = e^{z}$ can be used to observe that Eq.~(\ref{KLN}) reduces to the form factor for complete spectrum, see Eq.~(\ref{FF}).
\hfill $\blacksquare$
\end{remark}

\begin{figure}
\includegraphics[width=\textwidth]{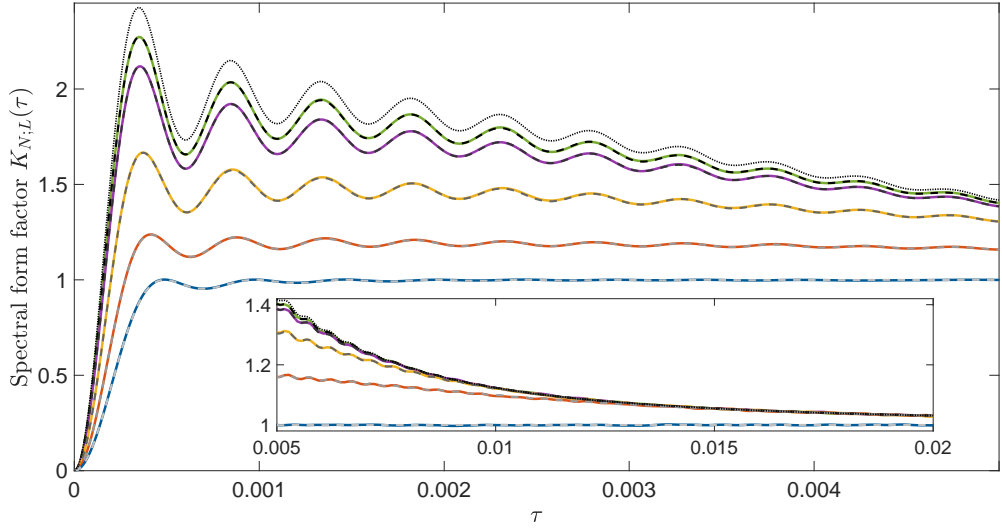}
\caption{Spectral form factor $K_{N;L}(\tau)$ as a function of $\tau$ for truncated eigenlevel sequences of the length $L=2048$ obtained from RDM spectra. Solid green, purple, yellow, red and blue curves correspond to the form factor simulated numerically for truncation parameters $\alpha_L=0.1$, $0.2$, $0.5$, $0.8$, and $1$. The number of realizations is $M=10^6$. The eigenvalues were uniformly distributed ${\rm U}(0, L/\alpha_L)$. Dashed lines on the top of solid colored ones display theoretical curves from Theorem~\ref{KNL-theorem}. Dotted line corresponds to the theoretical form factor for eigenlevel sequences of the length $L=2048$ with independent and exponentially distributed level spacings. Comparison with simulated curve can be found in Fig.~2 of Ref.~\cite{ROK-2020}. Inset: the same graphs for higher $\tau$.
}
\label{Plot_3}
\end{figure}

In Section~\ref{PS-RDM} (see Remark~\ref{rem-correlations} and Corollary~\ref{PS-inf-L-omega}) we have shown that the power spectrum of spectral subsequences of {\it finite} length $L$ drawn from infinite-dimensional random diagonal matrices coincides with the one for finite eigenvalue sequences with uncorrelated spacings addressed in Section~1.3 of Ref.~\cite{ROK-2020}. Not surprisingly, such a correspondence between two spectral models holds for their spectral form factors, too, see Remark~\ref{FFS-R} below.

Figure~\ref{Plot_3} demonstrates a perfect agreement between the analytical prediction of Theorem~\ref{KNL-theorem} and the form factor computed out of simulated RDM spectra.

\begin{theorem}[Finite subsequence]\label{K-inf-L-tau}
  Let $K_{N;L}(\tau)$ be a form factor of the spectral subsequence of length $L$ obtained from a complete spectrum of random diagonal matrices as specified in Theorem~\ref{KNL-theorem}. For $\tau={\mathcal O}(N^0)$ and $L={\mathcal O}(N^0)$, it holds that
\begin{eqnarray} \label{KL-tau}\fl \quad
    K_{\infty;L}(\tau) = \lim_{N\rightarrow \infty} K_{N;L}(\tau) \nonumber\\
    \fl\qquad\qquad\quad =
    1 + \frac{1}{4\pi^2\tau^2 L} \left(
    1 - \frac{1}{(1+4\pi^2\tau^2)^L}
    \right)   - \frac{1}{\pi\tau L} \frac{\sin[L \arctan(2\pi\tau)]}{(1+4\pi^2\tau^2)^{L/2}}.
\end{eqnarray}
\end{theorem}
\begin{proof}
Apply Proposition~\ref{1F1-as} to the hypergeometric functions appearing in Eq.~(\ref{KLN}).
\end{proof}
\begin{remark}\label{FFS-R}
Not surprisingly, the form factor $K_{\infty;L}(\tau)$ of spectral subsequences of a {\it finite} length $L$, produced by truncation of unbounded spectra of infinite-dimensional random diagonal matrices, coincides with the form factor in the model of eigenlevel sequences with uncorrelated, exponentially distributed spacings, see Section~1.3 of Ref.~\cite{ROK-2020}. Formally, this can be deduced by comparing Eq.~(\ref{KL-tau}) with Eq.~(1.22) of Ref.~\cite{ROK-2020}, in which $N$ is replaced with $L$ and the characteristing function $\Psi_s(\tau)$ of level spacings is set to $\Psi_s(\tau) = (1 - 2i\pi\tau)^{-1}$.
\hfill $\blacksquare$
\end{remark}

\begin{remark}
As $L\rightarrow \infty$, several scaling limits of the form factor $K_{\infty;L}(\tau)$ can be identified: for extremely short times $\tau={\mathcal O}(L^{-1})$, for intermediately small times $\tau = {\mathcal O}(L^{-1/2})$, and for finite times $\tau={\mathcal O}(L^0)$. Corresponding limiting laws are given by Eqs.~(\ref{KLN-1}), (\ref{KLN-2}) and (\ref{KLN-3}), respectuvely, of Theorem~\ref{KNL-infinite} upon {\it formal} setting $\alpha_L=0$ therein; see also Section~1.3 of Ref.~\cite{ROK-2020}.
\hfill $\blacksquare$
\end{remark}

To address the issue of validity of the form factor approximation raised in our previous papers~\cite{ROK-2017,ROK-2020}, we need to determine an asymptotic behavior of the spectral form factor $K_{N;L}(\tau)$ for truncated unfolded spectra of random diagonal matrices as $N \rightarrow \infty$. Various scaling limits of the form factor are established in Theorem~\ref{KNL-infinite} below. These large--$N$ results, combined with a similar analysis of the power spectrum presented in Section~\ref{PS-RDM}, will set the scene for discussion in Section~\ref{FFA-status}. The proof of Theorem~\ref{KNL-infinite} relies heavily on various asymptotic properties of the Kummer confluent hypergeometric function. Their detailed study is presented in Appendix~\ref{A-1}.

\begin{theorem}[Infinite subsequences]\label{KNL-infinite}
Let both the number $M=M(N)$ of skipped (low-lying) eigenlevels and the length $L=L(N)$ of the spectral subsequence scale with $N$ in such a way that the two limits exist
\begin{eqnarray}\label{scales-LMN}
    \alpha_L = \lim_{N\rightarrow\infty} \frac{L(N)}{N}, \quad \alpha_M = \lim_{N\rightarrow\infty} \frac{M(N)}{N},
\end{eqnarray}
where $L(N)+M(N)\le N$. Then, the form factor $K_{N;L}(\tau)$ admits the following scaling limits:
\begin{eqnarray}\label{KLN-1}
  K_{\alpha_L}^{(-1)}(T) &=&\lim_{N\rightarrow \infty} K_{N;L(N)}(\tau)\Big|_{\tau=T/L(N)} \nonumber\\
    &=& 2-\alpha_L - \alpha_L \frac{\sin^2(\pi T)}{(\pi T)^2} - 2 (1-\alpha_L) \frac{\sin(2\pi T)}{2\pi T}
\end{eqnarray}
and
\begin{eqnarray}\label{KLN-2}
K_{\alpha_L}^{(-\rfrac{1}{2})}({\mathcal T}) = \lim_{N\rightarrow \infty} K_{N; L(N)}(\tau)\Big|_{\tau={\mathcal T}/L(N)^{1/2}} \nonumber\\
            \qquad \qquad \qquad = 1+\frac{1}{4\pi^2 {\mathcal T}^2}\left( 1-e^{-4 \pi^2 {\mathcal T}^2(1-\alpha_L)} \right).
\end{eqnarray}
Furthermore, any other scaling limit of $K_{N;L}$ is trivial, i.e.~for $\tau=\mathfrak{T}/L(N)^{\delta}$ with fixed scaling exponent $\delta$, we have
\begin{eqnarray}\label{KLN-3}\fl
    \qquad K_{\alpha_L}^{(-\delta)}({\mathfrak T})=\lim_{N\rightarrow\infty} K_{N;L(N)}(\tau)\Big|_{\tau={\mathfrak T}/L(N)^{\delta}}=\left\{
\begin{array}{ll}
  1, & ${\rm if\;}$ \delta < \rfrac{1}{2}; \\
  2-\alpha_L, & ${\rm if\;}$ \rfrac{1}{2}<\delta <1.
\end{array}
\right.
\end{eqnarray}
\end{theorem}
\begin{proof}
(i) To prove Eq.~(\ref{KLN-1}), we start with the exact Eq.~(\ref{KLN}) of Theorem~\ref{KNL-theorem} and make use of Proposition~\ref{1F1-as} to deduce the asymptotic expansion
\begin{eqnarray} \fl \label{1F1-Th}
    \qquad {}_1 F_1 \left(\alpha_L N+j; N+j; \frac{2 i\pi T}{\alpha_L} \frac{N+1}{N}\right) \nonumber\\
\fl
  \qquad\qquad = e^{2 i \pi T}
    \left\{
        1 -  \frac{2 \pi T}{\alpha_L N} \Big( (1-\alpha_L) \pi T - i\, (j + (1-j)\alpha_L) \Big)
    + {\mathcal O}\left( \frac{1}{N^2} \right)
    \right\},
\end{eqnarray}
in which both leading and the first subleading in $1/N$ terms are kept. Substituting Eq.~(\ref{1F1-Th}) taken at $j=0$ and $j=-1$ into Eq.~(\ref{KLN}) produces Eq.~(\ref{KLN-1}).

(ii) To prove Eq.~(\ref{KLN-2}), we employ a small--$x$ version
\begin{eqnarray} \fl \label{ii-1}
   \qquad {}_1F_1(b \lambda+j;\lambda+j; i x \lambda)\ e^{-i\lambda b x} = \exp \left[ -\lambda \left( \frac{1}{2}b(1-b) x^2 + {\mathcal O}(x^3) \right) \right]\nonumber\\
   \fl
     \qquad \qquad\qquad\qquad \times \left(1 +{\mathcal O}(x)\right)\left( 1
     + {\mathcal O}\left( \frac{1}{\lambda} \right)
\right)
\end{eqnarray}
of Eq.~(\ref{pst-prop-main}) proven in Proposition~\ref{prop:1F1new} to observe that
\begin{eqnarray} \fl \label{1F1-Th-ii}
    \qquad {}_1 F_1 \left(\alpha_L N+j; N+j; \frac{2 i\pi \mathcal{T}\sqrt{N}}{\sqrt{\alpha_L}} \frac{N+1}{N}\right) \, \exp\left(
        -2 i \pi \mathcal{T} \sqrt{\frac{\alpha_L}{N}} (N+1)
        \right)
     \nonumber\\
        \qquad\qquad = e^{-2 \pi^2 \mathcal{T}^2 (1-\alpha_L) }  \left( 1
     + {\mathcal O}\left( \frac{1}{\sqrt{N}} \right)\right)
\end{eqnarray}
as $N \rightarrow \infty$. Notice that r.h.s.~in~Eqs.~(\ref{ii-1}) and (\ref{1F1-Th-ii}) does not depend on $j$. Substitution of Eq.~(\ref{1F1-Th-ii}) into Eq.~(\ref{KLN}) brings the sought Eq.~(\ref{KLN-2}).

(iii) Now we are in position to prove Eq.~(\ref{KLN-3}).

(iii-a) For $\delta < \rfrac{1}{2}$, we realize that the hypergeometric functions in Eq.~(\ref{KLN}) are bounded, as spelt out by Lemma~\ref{lemma:bound}; hence, the $N\rightarrow \infty$ limit of Eq.~(\ref{KLN}) is governed by the product $\tau^2 N$. If $\tau = {\mathcal O}(N^{-\delta})$, the product $\tau^2 N ={\mathcal O}(N^{1-2\delta})$ grows unboundedly for all $\delta<\rfrac{1}{2}$ thus producing unity in the first line of Eq.~(\ref{KLN-3}).

(iii-b) For $\rfrac{1}{2} < \delta < 1$, Proposition~\ref{prop:1F1new} (see also Eq.~(\ref{ii-1})) implies that
\begin{eqnarray}
\fl \label{1F1-Th-b}
    \qquad {}_1 F_1 \left(\alpha_L N+j; N+j; \frac{2 i\pi \mathfrak{T}}{(N\alpha_L)^\delta}(N+1)\right)
    \exp\left(
        - 2 i\pi  \mathfrak{T} (N\alpha_L)^{1-\delta}\frac{N+1}{N}
    \right)
 \nonumber\\
\fl
  \qquad\qquad\qquad\qquad = 1 - \frac{2\pi^2 \mathfrak{T}^2 (1-\alpha_L)}{(N \alpha_L)^{2\delta-1}} + {\mathcal O}(N^{-\delta}) + {\mathcal O}(N^{2-4\delta}).
\end{eqnarray}
Substitution of Eq.~(\ref{1F1-Th-b}) into Eq.~(\ref{KLN}) yields the second line of Eq.~(\ref{KLN-3}).

\end{proof}

\begin{remark}
  Equation~(\ref{KLN-3}) of Theorem~\ref{KNL-infinite}, taken at $\delta=0$, brings out an asymptotic behavior of spectral form factor in the bulk $\tau ={\mathcal O(L(N)^0)}$ as $N\rightarrow \infty$:
\begin{eqnarray} \label{K=1}
    \lim_{N\rightarrow\infty} K_{N;L(N)}(\tau) = 1.
\end{eqnarray}
\hfill $\blacksquare$
\end{remark}
\begin{remark}
Notice the limits
\begin{eqnarray} \label{Kc-00}
    \lim_{T\rightarrow 0} K_{\alpha_L}^{(-1)}(T) &=&\lim_{\tau\rightarrow 0} K_{N;L}(\tau)=0,\\
        \label{Kc-23}
    \lim_{\mathcal{T}\rightarrow \infty} K_{\alpha_L}^{(-\rfrac{1}{2})}(\mathcal{T})&=&
\lim_{\tau\rightarrow\infty} K_{N;L}(\tau) = K_{\alpha_L}^{(0)}(\tau) \equiv 1
\end{eqnarray}
and the continuity of the scaling curves in Theorem~\ref{KNL-infinite}:
\begin{eqnarray} \label{Kc-12}
\lim_{T\rightarrow \infty} K_{\alpha_L}^{(-1)}(T) =\lim_{\mathcal{T}\rightarrow 0} K_{\alpha_L}^{(-\rfrac{1}{2})}(\mathcal{T})=2-\alpha_L.
\end{eqnarray}
\hfill $\blacksquare$
\end{remark}

The scaling curves derived in Theorem~\ref{KNL-infinite} are displayed in Fig.~\ref{Plot_4}, where the three curves for the spectral form factor -- $K^{(-1)}_{\alpha_L}(T)$, $K^{(-\rfrac{1}{2})}_{\alpha_L}({\mathcal T})$ and $K^{(0)}_{\alpha_L}(\tau)$ -- are glued together across all three scaling regimes. Continuity of the resulting curve at two borderlines -- between the scaling regimes (I) and (II), and (II) and (III) -- are secured by Eqs.~(\ref{Kc-12}) and (\ref{Kc-23}), respectively. We shall return to Fig.~\ref{Plot_4} in Section~\ref{FFA-status}.

\begin{figure}
\includegraphics[width=\textwidth]{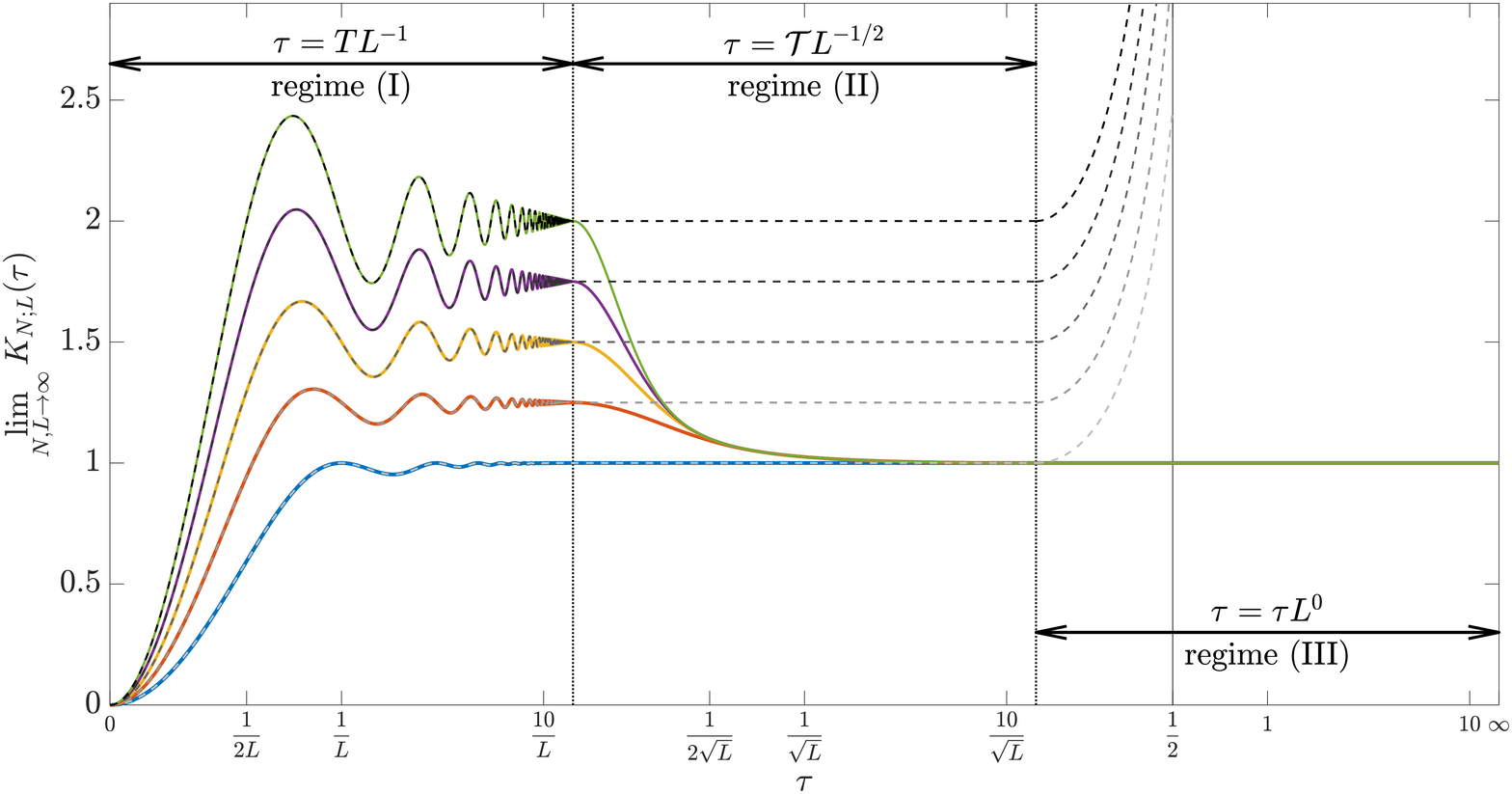}
\caption{
Limiting curves ($N\rightarrow \infty$) for the form factor across three scaling regimes [${\rm (I)}$ -- Eq.~(\ref{KLN-1}), ${\rm (II)}$ -- Eq.~(\ref{KLN-2}), and ${\rm (III)}$ -- Eq.~(\ref{KLN-3}) at $\delta=0$], glued together at vertical dotted lines. The functions $K^{(-1)}_{\alpha_L}(T)$, $K^{(-1/2)}_{\alpha_L}({\mathcal T})$ and
$K^{(0)}_{\alpha_L}(\tau)$, describing the regimes ${\rm (I)}$, ${\rm (II)}$ and ${\rm (III)}$, correspondingly, are plotted vs variables $T=L\tau$, ${\mathcal T}=L^{1/2}\tau$ and $\tau$, each running over the entire real half-line compactified using the transformation $(0,\infty)=\tan((0,\pi/2))$. Solid blue, red, yellow, purple and green curves correspond to the form factor for truncated RDM sequences with truncation parameters $\alpha_L =1, 0.75, 0.5, 0.25, 0$. The dashed grey lines display the limiting curve for the regularized power spectrum $\lim_{N\rightarrow \infty}\omega^2 S_{N;L}(\omega)$ with $0\le \omega=2\pi\tau \le \pi$ (that is, $0\le \tau \le \rfrac{1}{2}$). In the scaling regimes ${\rm (I)}$, ${\rm (II)}$ and ${\rm (III)}$, the curve is described by Eq.~(\ref{S-minus-1}), Eq.~(\ref{S-minus-lambda}) taken at $\delta=1/2$, and Eq.~(\ref{S-0}), respectively.
}
\label{Plot_4}
\end{figure}

\subsection{Form factor approximation: Power spectrum vs form factor} \label{FFA-status}

Equipped with the nonperturbative results for the power spectrum and the form factor for the exactly solvable model of random diagonal matrices, we are in position to examine a status of the so-called form factor approximation, introduced heuristically in Ref.~\cite{FGMMRR-2004}. Assuming that the power spectrum of a quantum system is merely determined by its two-point spectral correlations, the authors of Ref.~\cite{FGMMRR-2004} claimed to relate the power spectrum $S(\omega)$ of a quantum system to its spectral form factor $K(\tau)$ in the {\it entire} frequency domain $0 \le \omega \le \pi$. Referring interested readers to Eqs.~(3), (8) and (10) of the original paper Ref.~\cite{FGMMRR-2004}, we only quote a small-$\omega$ reduction of their result:
\begin{eqnarray} \label{FFA-heuristics}
    S(\omega)\Big|_{\omega \ll 1} \approx \frac{1}{\omega^2} K\left(
    \tau = \frac{\omega}{2\pi}
    \right),
\end{eqnarray}
see also Ref.~\cite{PRPAG-2018}.
\newline\newline\noindent
{\it Complete spectrum.}---In the case of complete spectrum (truncation parameter $\alpha_L=1$), two scaling regimes appear in both the regularized power spectrum and the form factor as $N\rightarrow \infty$. For extremely small frequencies $\omega=\Omega/N$ and very short times $\tau=T/N$, Eq.~(\ref{S-minus-1}) of Corollary~\ref{PS-SL-trunc} and Eq.~(\ref{KLN-1}) of Theorem~\ref{KNL-infinite} bring
\begin{eqnarray}
     \lim_{N\rightarrow\infty} \omega^2 S_N(\omega)\Big|_{\omega=\Omega/N} = 1 - \frac{\sin^2(\Omega/2)}{(\Omega/2)^2}
\end{eqnarray}
and
\begin{eqnarray}
     \lim_{N\rightarrow\infty} K_N(\tau)\Big|_{\tau=T/N} = 1 - \frac{\sin^2(\pi T)}{(\pi T)^2},
\end{eqnarray}
correspondingly. This immediately implies that the heuristic equality Eq.~(\ref{FFA-heuristics}) does hold in the domain $\omega = {\mathcal O}(N^{-1})$ and $\tau = {\mathcal O}(N^{-1})$, upon identifying $T=\Omega/2\pi$. Remarkably, this equality stays valid also for finite albeit small frequencies $\omega =o(N^0)$ and times $\tau = o(N^0)$ as both $\omega^2 S_N(\omega)$ and $K_N(\tau)$ approach unity as $N\rightarrow \infty$, see Eq.~(\ref{S-0}) taken at $\alpha_L=1$ and Eq.~(\ref{K=1}). Hence, our rigorous analysis {\it does} validate the form factor approximation Eq.~(\ref{FFA-heuristics}) for {\it complete spectra of random diagonal matrices} up to $\omega=o(N^0)$.
\newline\newline\noindent
{\it Truncated spectrum.}---As far as the truncated spectrum is concerned (truncation parameter $0\le \alpha_L<1$), three scaling regimes can be identified for the power spectrum and the form factor.

(i) In the first regime, referring to frequencies $\omega=\Omega/L(N)$ and times $\tau = T/L(N)$,
Eq.~(\ref{S-minus-1}) of Corollary~\ref{PS-SL-trunc} and Eq.~(\ref{KLN-1}) of Theorem~\ref{KNL-infinite} yield
\begin{eqnarray} \label{FSE-S}
     \mathcal{S}_{\alpha_L}^{(-1)}(\Omega)&=& \lim_{N\rightarrow\infty} \omega^2 S_{N;L(N)}(\omega)\Big|_{\omega=\Omega/L(N)} \nonumber\\
        &=& 2 - \alpha_L -\alpha_L \frac{\sin^2(\Omega/2)}{(\Omega/2)^2} -2(1-\alpha_L)\frac{\sin\Omega}{\Omega}
\end{eqnarray}
and
\begin{eqnarray} \label{FSE-K}
    K_{\alpha_L}^{(-1)}(T) &=& \lim_{N\rightarrow\infty} K_{N;L(N)}(\tau)\Big|_{\tau=T/L(N)} \nonumber\\
        &=& 2 - \alpha_L -\alpha_L \frac{\sin^2(\pi T)}{(\pi T)^2} -2(1-\alpha_L)\frac{\sin(2\pi T)}{2\pi T},
\end{eqnarray}
respectively. Since the two expressions coincide with each other upon substitution $T=\Omega/2\pi$, we conclude that a heuristic form factor approximation is indeed justified close to the origin, that is, in the domain $\omega={\mathcal O}(L(N)^{-1})$ and $\tau={\mathcal O}(L(N)^{-1})$.

(ii) In the second regime, corresponding to higher frequencies $\omega={\mathcal O}(L(N)^{-1/2})$ and larger times $\tau={\mathcal O}(L(N)^{-1/2})$, Eq.~(\ref{S-minus-lambda})
of Corollary~\ref{PS-SL-trunc} (with $\delta$ set to $\rfrac{1}{2}$) and Eq.~(\ref{KLN-2}) of Theorem~\ref{KNL-infinite} yield two functionally different limiting laws for the power spectrum
\begin{eqnarray}\label{LimS-2}
    \mathcal{S}_{\alpha_L}^{(-\rfrac{1}{2})}(\tilde{\Omega})=\lim_{N\rightarrow\infty} \omega^2 S_{N;L(N)}(\omega)\Big|_{\omega = \tilde{\Omega}/L(N)^{1/2}} = 2-\alpha_L,
\end{eqnarray}
and the form factor
\begin{eqnarray}\label{LimK-2} \fl
   K_{\alpha_L}^{(-\rfrac{1}{2})}(\mathcal{T}) = \lim_{N\rightarrow \infty} K_{N; L(N)}(\tau)\Big|_{\tau={\mathcal T}/L(N)^{1/2}}
            = 1+\frac{1}{4\pi^2 {\mathcal T}^2}\left( 1-e^{-4 \pi^2 {\mathcal T}^2(1-\alpha_L)} \right).
\end{eqnarray}
The two limits Eqs.~(\ref{LimS-2}) and (\ref{LimK-2}) are seen to coincide only approximately for ${\mathcal T} \ll (1-\alpha_L)^{-1/2}$, quickly deviating from what is `predicted' by the form factor approximation Eq.~(\ref{FFA-heuristics}). This signals of the breakdown of the form factor approximation in the domain of frequencies and times of order ${\mathcal O}(L(N)^{-1/2})$.

(iii) In the third regime, referring to both finite frequencies $\omega={\mathcal O}(L(N)^0)$ and times $\tau={\mathcal O}(L(N)^0)$ further restricted to the domain $\omega \ll 1$ and $\tau \ll 1$ (as imposed by Eq.~(\ref{FFA-heuristics})), Eqs.~(\ref{S-0}) and (\ref{K=1}) furnish
\begin{eqnarray}\label{LimS-1}
    \mathcal{S}_{\alpha_L}^{(0)}(\omega)=\lim_{N\rightarrow\infty} \omega^2 S_{N;L(N)}(\omega)\Big|_{\omega = {\mathcal O}(L(N)^0) \ll 1} = 2-\alpha_L + {\mathcal O}(\omega^2)
\end{eqnarray}
and
\begin{eqnarray}\label{LimK-1}
    K_{\alpha_L}^{(0)}(\tau)=\lim_{N\rightarrow\infty} K_{N;L(N)}(\tau)\Big|_{\tau = {\mathcal O}(L(N)^0) \ll 1} = 1.
\end{eqnarray}
As the two limits Eqs.~(\ref{LimS-1}) and (\ref{LimK-1}) differ from each other for all $0\le \alpha_L <1$, we conclude that, in the domain of finite frequencies $\omega \ll 1$ and times $\tau \ll 1$, the form factor approximation is broken down, too, for any truncated eigenlevel sequence. The heavier the truncation the higher discrepancy is; the discrepancy factor $2-\alpha_L$ reaches its maximum $2$ for heavily truncated sequences ($\alpha_L=0$) describing spectral fluctuations in realistic quantum systems with integrable classical dynamics, see Discussion in Section~\ref{PS-RDM}.

The above consideration is visualized in Fig.~\ref{Plot_4}, where we compare the $N\rightarrow\infty$ limiting curves for the regularized power spectrum and for the form factor, in all three scaling regimes. While both spectral indicators are on the top of each other in the first scaling regime (characterized by extremely low frequencies $\omega={\mathcal O}(L^{-1})$ and short times $\tau={\mathcal O}(L^{-1})$), they start to deviate from each other already at $\omega={\mathcal O}(L^{-1/2})$ and $\tau={\mathcal O}(L^{-1/2})$, except for $\alpha_L = 1$ (the case of complete spectrum). The deviations keep growing with further increase of $\omega$ and $\tau$. Figure~\ref{Plot_5} demonstrates that a similar behavior is also observed for RDM spectra of a large but finite size.

\begin{figure}
\includegraphics[width=\textwidth]{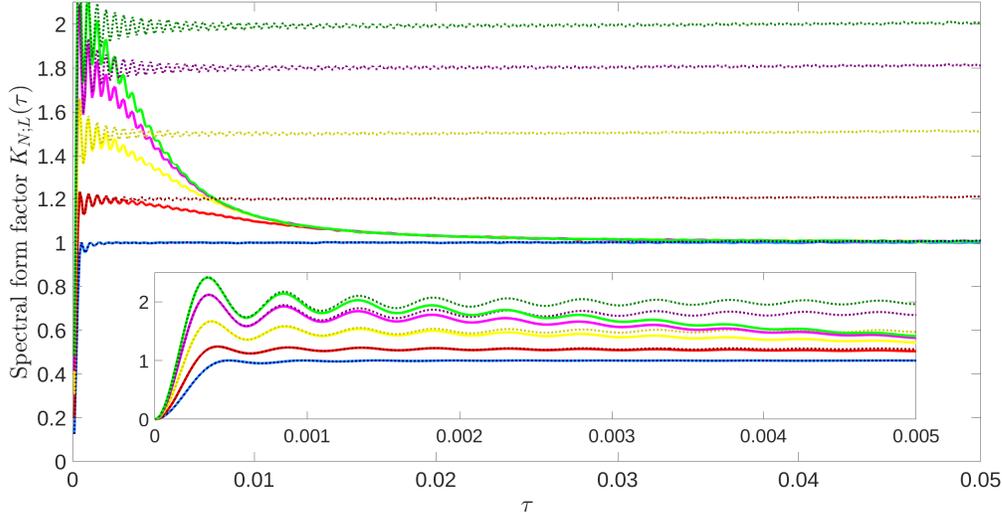}
\caption{
Spectral form factor $K(\tau)$ (solid lines) and regularized power spectrum $\omega^2 S(\omega)$ (dotted lines, $\omega=2\pi \tau$) simulated for truncated eigenlevel sequences of the length $L=2048$ obtained from RDM spectra. Green, purple, yellow, red and blue curves correspond to truncation parameters $\alpha_L = 0.1$, $0.2$, $0.5$, $0.8$, and $1$. The number of realizations is $M=10^6$. Inset: magnified view of the same graphs.
}
\label{Plot_5}
\end{figure}

To conclude, our analysis shows that the heuristic relation Eq.~(\ref{FFA-heuristics}), considered in the context of {\it infinite-dimensional random diagonal matrices},  holds true for their {\it complete} spectra up to $\omega={\mathcal O}(L^0)$ as long as $\omega \ll 1$. However, for {\it truncated spectra}~\footnote[4]{This includes the case $\alpha_L=0$ which corresponds to a model of random spectra with uncorrelated spacings discussed in Ref.~\cite{ROK-2020}.}, the form factor approximation breaks down much earlier, at the very beginning of the scaling regime (II), see Fig.~\ref{Plot_4}. Hence, the form factor approximation is restricted to very low frequencies $\omega=o(L^{-1/2})$ which give no access to the bulk of the power spectrum.

\fancyhead{} \fancyhead[RE,LO]{\textsl{\thesection \quad Numerical study of power spectrum}}
\fancyhead[LE,RO]{\thepage}
\section{Numerical study of power spectrum in quantum systems with integrable classical dynamics} \label{Num-Billiards}

In numerical studies of spectral observables of {\it individual} quantum systems, an energy averaging, rather than an ensemble averaging, is usually involved. Specifically, an observable of interest is computed along a spectral segment of the length $L$ and then averaged over various locations of such segments in the spectrum~\cite{CCG-1985}. In the context of our study, such an energy averaging is straightforward to define for both the power spectrum
\begin{eqnarray}\label{ps-def-numerics}\fl
    \qquad S_L(\omega) = \frac{1}{L} \left< \sum_{\ell, m= 1}^{L} (E_{M_0+\ell} - E_{M_0})
                    (E_{M_0+m} - E_{M_0}) e^{i\omega (\ell-m)} \right>_{M_0}
\end{eqnarray}
and the form factor
\begin{eqnarray}\label{sff-def-numerics}\fl
    \qquad K_L(\tau) = \frac{1}{L} \Bigg\{ \left< \sum_{\ell, m= 1}^{L} e^{2i\pi \tau (E_{M_0+\ell} - E_{M_0})}
                     e^{-2i\pi \tau(E_{M_0+m} - E_{M_0})} \right>_{M_0} \nonumber\\
                     - \left< \sum_{\ell= 1}^{L} e^{2i\pi \tau(E_{M_0+\ell} - E_{M_0})} \right>_{M_0}
                     \left< \sum_{m=1}^L e^{-2i\pi \tau(E_{M_0+m} - E_{M_0})} \right>_{M_0} \Bigg\}.
\end{eqnarray}
Here, $\{ E_1 \le E_2 \le \dots \}$ is a set of {\it deterministic}, ordered and unfolded eigenlevels of an individual quantum system; angular brackets
$\langle (\dots) \rangle_{M_0}$ stand for energy averaging performed by sampling spectral segments of the length $L$ at various locations $M_0$ along unbounded spectrum. Following a detailed discussion of collective fluctuations in Section~\ref{PSSF}, both Eqs.~(\ref{ps-def-numerics}) and (\ref{sff-def-numerics}) deal with {\it tuned} spectral segments. To capture the {\it universal aspects} of spectral statistics expected to emerge for high-lying eigenlevels, the number of skipped eigenlevels $M_0$ should be large enough (see below).

\subsection{Semicircular billiards}\label{SCBN}

To study the power spectrum in a circular billiard of the radius $R$, we start with the Schr\"odinger equation
\begin{eqnarray}\label{c-SE}
    \left( \frac{\partial^2}{\partial r^2} + \frac{1}{r}\frac{\partial}{\partial r} + \frac{1}{r^2}\frac{\partial^2}{\partial \theta^2}\right)\Psi_{n,m}(r,\theta) + E_{n,m} \Psi_{n,m}(r,\theta) = 0,
\end{eqnarray}
where we set $\hbar^2/2m=1$, and solve it the domain ${\mathcal D}_{\rm\ocircle}=\{ 0\le r\le R; -\pi \le \theta \le \pi\}$ assuming Dirichlet boundary conditions
$\Psi_{n,m}((r,\theta)\in \partial {\mathcal D}_{\rm\ocircle})=0$. Its solution brings eigenfunctions and eigenspectrum in the form
\begin{eqnarray}\label{c-SE-functions}
    \Psi_{n,m}(r,\theta) = e^{i m\theta} J_m\left( \kappa_{n,m}\frac{r}{R}\right)
\end{eqnarray}
and
\begin{eqnarray}\label{c-SE-spectrum}
    E_{n,m} =  \frac{\kappa_{n,m}^2}{R^2}.
\end{eqnarray}
Here, $\kappa_{n,m}$ denotes the $n$-th positive zero ($n=1,2,\dots$) of the $m$-th Bessel function, $J_{m}(\kappa_{n,m})=0$ with $m=0,\pm 1,\pm 2,\dots$.

Owing to the identity $\kappa_{n,-m}=\kappa_{n,m}$, the energy spectrum Eq.~(\ref{c-SE-spectrum}) is double degenerate for all $m \neq 0$. To remove such a degeneracy, we follow the idea of Ref.~\cite{RV-1998} and consider a semicircular billiard comprised by the domain ${\mathcal D}_{\rm\Leftcircle}=\{ 0\le r\le R; 0 \le \theta \le \pi\}$. For the Dirichlet boundary conditions, the eigenstates of ${\mathcal D}_{\rm\Leftcircle}$
\begin{eqnarray}
    \chi_{n,m}(r,\theta) = \sin(m\theta)\, J_m\left( \kappa_{n,m}\frac{r}{R}\right), \qquad n,m=1,2,\dots,
\end{eqnarray}
are antisymmetric with respect to the reflection across the line $\theta=0$ and the eigenspectrum is given by Eq.~(\ref{c-SE-spectrum}) with $n,m=1,2,\dots$.

The large--$E$ asymptotics of spectral counting function $\bar{N}(E)$, describing the number of eigenlevels below energy $E$, is given by the Weyl law~\cite{BH-1976}
\begin{eqnarray}\label{Weyl}
    \bar{N}(E) = \frac{{\mathcal A}}{4\pi} E - \frac{\mathcal L}{4\pi} \sqrt{E} + {\mathcal C}, \quad E \rightarrow \infty,
\end{eqnarray}
where ${\mathcal A}=\pi R^2/2$ is the semicircle area, $\mathcal{L}=(2+\pi)R$ is its circumference, and ${\mathcal C}$ is a curvature-and-corner's contribution whose precise value is known but irrelevant for our purposes. In what follows we set $R=2 \sqrt{2}$ to ensure that the unit mean spacing $\Delta = (d\bar{N}/dE)^{-1}=1$ as $E \rightarrow \infty$.

In concert with a discussion in the beginning of Section~\ref{Num-Billiards}, the following computational protocol will be used to compute the form factor and power spectrum in the semicircular billiard.\newline\newline
\noindent
{\bf Protocol 1.}~To generate an ensemble ${\mathcal E} = \{ {\mathcal E}_1,\dots, {\mathcal E}_Q\}$ of $Q$ eigenlevel sequences ($Q\gg 1$ is required for high-accuracy computation of the form factor and the power spectrum), we produce a sufficiently long, ordered, eigenlevel sequence $\{E_1 \le E_2 \le \dots \}$ for a semicircular billiard of the radius $R=2\sqrt{2}$. To attain the universal spectral regime, we choose a high-lying threshold energy $E_*> 10^{9}$ (measured in the units of mean level spacing $\Delta=1$) and count the number $M(E_*)$ of eigenvalues {\it below} the threshold $E_*$; all of them will be skipped. The first truncated spectral subsequence $\mathcal{E}_1$ of the length $L$ belonging to the ensemble ${\mathcal E}$ is set to be $\mathcal{E}_1 =
    \{ E_{M(E_*)+1}-E_{M(E_*)}, \dots, E_{M(E_*)+L}- E_{M(E_*)}\}$,
notice the tuning! To produce the remaining $(Q-1)$ spectral subsequences of the ensemble ${\mathcal E}$, we generate a set of $(Q-1)$ independent discrete random variables $\{n_1, n_2,\dots, n_{Q-1}\}$, with $n_k$ denoting a random gap between the last eigenvalue of the $k$-th and the first eigenvalue of the $(k+1)$-th truncated subsequence of the same length $L$. The $j$-th member of the ensemble ${\mathcal{E}}$ is comprised by truncated subsequence of the form
\begin{eqnarray} \fl
    \quad \mathcal{E}_j =
    \{ E_{M(E_*)+(j-1)L + \sum_{k=1}^{j-1} n_{k} +1 } -
        E_{M(E_*)+(j-1)L + \sum_{k=1}^{j-1} n_{k} }, \dots, \nonumber\\
    \fl \qquad \qquad \quad
    E_{M(E_*)+(j-1)L + \sum_{k=1}^{j-1} n_{k} + L } -
        E_{M(E_*)+(j-1)L + \sum_{k=1}^{j-1} n_{k} }\}.
\end{eqnarray}
The averaging procedure is then performed over the ensemble ${\mathcal E} = \{ {\mathcal E}_1,\dots, {\mathcal E}_Q\}$ of $Q\gg 1$ spectral sequences, each of length $L$.
\hfill $\blacksquare$
\newline\newline
This protocol was applied for numerical studies of the power spectrum and the spectral form factor in a semicircular billiard of the radius $R=2\sqrt{2}$ at energies larger than the threshold energy $E_* =10^{11}$. The spectral data were gathered by generating $Q=10^5$ spectral sequences of length $L=1024$ each, separated by independent, uniformly distributed random gaps $\{n_1, n_2,\dots,n_{Q-1}\}$, with~\footnote[2]{We have checked that the results are largely insensitive to a particular choice of the gaps.} $n_j \thicksim {\rm U}(50,150)$. The results are presented in Fig.~\ref{Plot_6}.

(i) Spectral form factor and power spectrum computed numerically for semicircular billiards show excellent agreement with the theoretical predictions derived for {\it heavily truncated spectra} of infinite-dimensional random diagonal matrices (Theorem~\ref{K-inf-L-tau} and Corollary~\ref{PS-inf-L-omega}). In accordance with the discussion in Sections~\ref{PS-RDM} and \ref{SFF-sec}, this sector ($\alpha_L=0$) of the theory of RDM spectra is equivalent to a spectral model of finite eigenlevel sequences with independent, exponentially distributed level spacings, see also Ref.~\cite{ROK-2020}.

(ii) For comparison, we have also plotted theoretical curves for the form factor and the power spectrum calculated for {\it complete spectra} of random diagonal matrices (see Propositions~\ref{RDM-PS-complete} and \ref{RDM-FF-complete}). It appears that the two have nothing in common with the corresponding statistics in semicircular billiards thus invalidating a direct use of random diagonal matrices to description of spectral fluctuations in realistic quantum systems with integrable classical dynamics.

(iii) As a side remark, we note that while the power spectrum in truncated and complete RDM spectra are nowhere the same, their form factors approach unity in both models as $\tau$ grows; see Fig.~\ref{Plot_4} for further details.

\begin{figure}
\includegraphics[width=\textwidth]{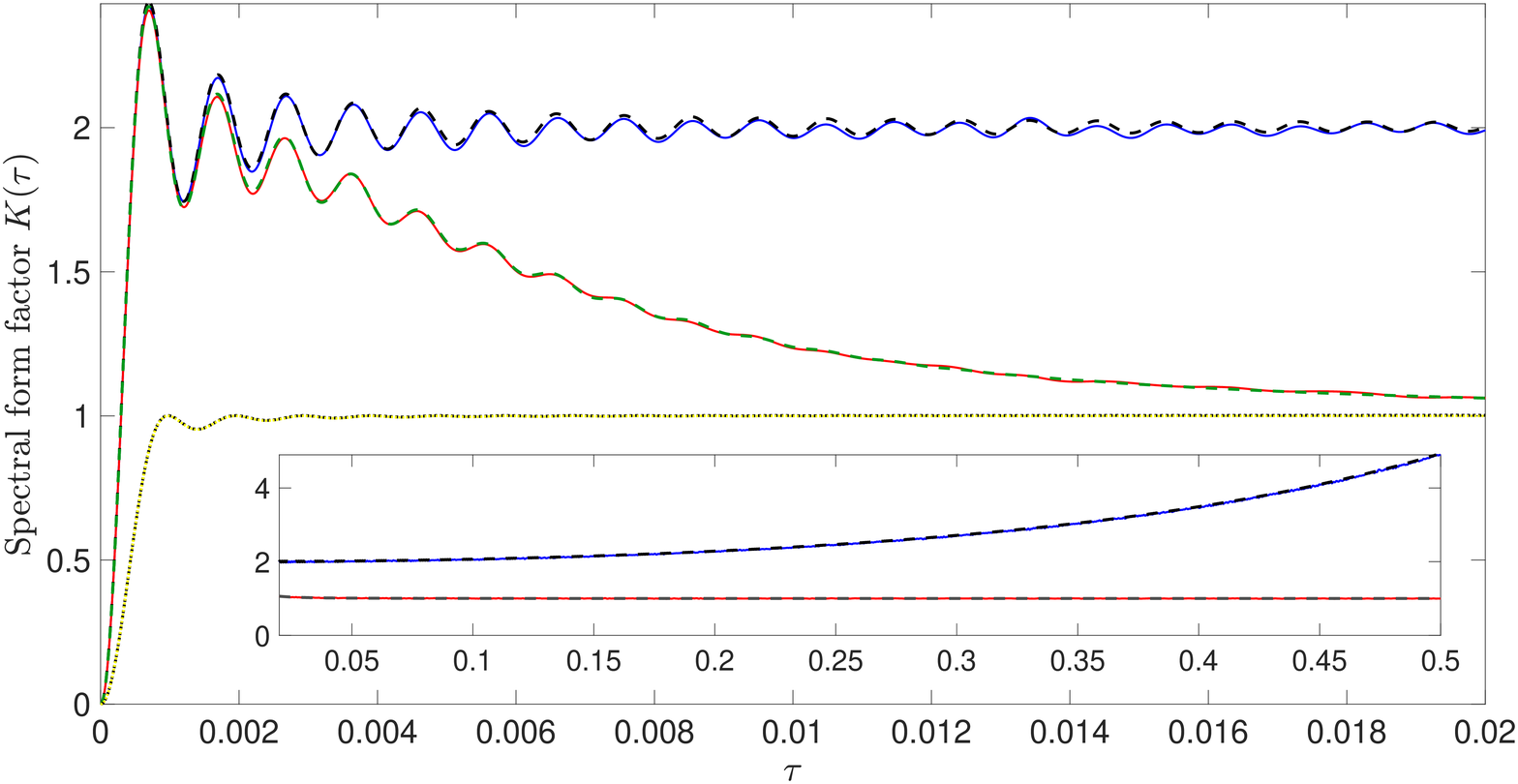}
\caption{
Spectral form factor $K(\tau)$ (red solid line) and regularized power spectrum $\omega^2 S(\omega)$ (blue solid line, $\omega=2\pi \tau$) calculated numerically for a semicircular billiard as described in Protocol 1. Parameters used: $R=2\sqrt{2}$, $E_*=10^{11}$, $L=1024$, and $Q=10^5$. Green and black dashed lines are the theoretical curves for the form factor [Eq.~(\ref{KL-tau})] and regularized power spectrum [Eq.~(\ref{SL-omega})], respectively. Both theoretical curves correspond to heavily truncated spectra of infinite-dimensional RDM which is equivalent to a spectral model of finite eigenlevel sequences with independent and exponentially distributed level spacings, see Remarks~\ref{FFS-R} and~\ref{rem-correlations}. Inset: the same graphs shown for larger $\tau$. For comparison, yellow solid and black dotted lines show $K_N(\tau)$ and $\omega^2 S_N(\omega)$ for complete spectra of $1024 \times 1024$ random diagonal matrices, see Eqs.~(\ref{FF}) and (\ref{smd-sum-33}).
}
\label{Plot_6}
\end{figure}

\subsection{Irrational rectangular billiards} \label{IRBN}

In the case of a rectangular billiard with sides $a$ and $b$, we start with the Schr\"odinger equation
\begin{eqnarray}
    \left( \frac{\partial^2}{\partial x^2} + \frac{\partial^2}{\partial y^2}\right)\Psi_{n,m}(x,y) + E_{n,m} \Psi_{n,m}(x,y) = 0,
\end{eqnarray}
where we set $\hbar^2/2m=1$, and solve it in the domain ${\mathcal D}=\{ 0\le x\le a; 0\le y\le b\}$ assuming Dirichlet boundary conditions
$\Psi_{n,m}((x,y)\in \partial {\mathcal D})=0$. Its solution brings both eigenfunctions and eigenspectrum in the form
\begin{eqnarray}
    \Psi_{n,m}(x,y) = \sin\left(\frac{\pi n x}{a}\right) \sin\left(\frac{\pi m y}{b}\right)
\end{eqnarray}
and
\begin{eqnarray}
    E_{n,m} = \pi^2 \left[
        \left(\frac{n}{a}\right)^2 +  \left(\frac{m}{b}\right)^2
    \right], \quad n,m = 1, 2, \dots
\end{eqnarray}
The spectral counting function $\bar{N}(E)$, describing the number of eigenlevels below energy $E$, is asymptotically given by the Weyl law Eq.~(\ref{Weyl}), where
${\mathcal A}=ab$ is the billiard area, $\mathcal{L}=2(a+b)$ is its circumference, and ${\mathcal C}$ is a corner's contribution.

In what follows we set ${\mathcal A}=4\pi$ to ensure the unit mean spacing $\Delta = (d\bar{N}/dE)^{-1}=1$ as $E \rightarrow \infty$. To focus
on generic, {\it universal} aspects of spectral statistics, we choose the squared aspect ratio
\begin{eqnarray} \label{alpha-irr}
    \alpha = \left(\frac{a}{b}\right)^2 \in {\mathbb R}_+ \backslash {\mathbb Q}
\end{eqnarray}
to be an irrational number hereby suppressing contributions of non-universal accidental spectral degeneracies~\cite{BT-1977, CK-1999}. This leaves us with the energy levels parametrized as follows:
\begin{eqnarray} \label{Enm-alpha}
    E_{n,m} = \frac{\pi}{4} \left[
        \frac{n^2}{\sqrt{\alpha}} + m^2\sqrt{\alpha}
    \right], \quad n,m = 1, 2, \dots.
\end{eqnarray}

To compute the form factor and the power spectrum with high accuracy, one has to create an {\it ensemble} of eigenlevel sequences and appropriately perform ensemble averaging. Contrary to the numerical study of a semicircular billiard, where only one protocol was devised, this can be achieved in several ways.

Below we discuss three different computational protocols particularly suitable for numerical studies of rectangular billiards. In Protocol 1, statistical ensemble arises as a result of samplings of a sufficiently long deterministic eigenlevel sequence `measured' for an {\it individual billiard}, see the discussion of `energy averaging' in the beginning of Section~\ref{Num-Billiards}. The Protocols 2 and 3 are inapplicable to studies of an individual billiard; instead, the two are tailor-made for statistical analysis of an {\it ensemble of billiards} with different aspects ratios but asymptotically equivalent mean spacings. Choosing the ratios at random introduces a true randomness similar to the one encountered in our analysis of ensembles of random diagonal matrices. Computation of the form factor and the power spectrum according to Protocol 2 is in one-to-one correspondence with the tuning procedure discussed and implemented in Section~\ref{PSSF}. The Protocol 3, describing a somewhat heuristic way of dealing with collective fluctuations, appears to be more effective numerically-wise.
\newline\newline
{\bf Protocol 1.}~To generate an ensemble ${\mathcal E}^{\prime} = \{ {\mathcal E}_1^{\prime},\dots, {\mathcal E}_Q^{\prime}\}$ of $Q$ eigenlevel sequences, we shall sample a sufficiently long, ordered, eigenlevel sequence $\{E_1^{(\alpha_0)} \le E_2^{(\alpha_0)} \le \dots \}$ produced for a fixed value of the squared aspect ratio that we set to the golden ratio $\alpha_0=(1+\sqrt{5})/2$. Similarly to the Protocol 1 for the semicircular billiard, we choose a high-lying threshold energy $E_*>10^{9}$ to attain the universal spectral regime and count the number $M(E_*;\alpha_0)$ of eigenvalues {\it below} the threshold $E_*$; all of them will be skipped. The first truncated spectral subsequence $\mathcal{E}_1^{\prime}$ of the length $L$ belonging to the ensemble ${\mathcal E}^{\prime}$ is chosen to be
\begin{eqnarray} \fl \qquad\quad
\mathcal{E}_1^{\prime} =
    \{ E_{M(E_*;\alpha_0)+1}^{(\alpha_0)}-E_{M(E_*;\alpha_0)}^{(\alpha_0)}, \dots, E_{M(E_*;\alpha_0)+L}^{(\alpha_0)}- E_{M(E_*;\alpha_0)}^{(\alpha_0)}\},
\end{eqnarray}
notice the tuning! Remaining $(Q-1)$ spectral subsequences of the ensemble ${\mathcal E}^{\prime}$ will be produced by generating a set of $(Q-1)$ independent random gaps $\{n_1, n_2,\dots, n_{Q-1}\}$, where $n_k$ denotes a random gap between the last eigenvalue of the $k$-th and the first eigenvalue of the $(k+1)$-th truncated subsequence of the same length $L$. The $j$-th member of the ensemble ${\mathcal{E}^{\prime}}$ is then comprised by truncated subsequence of the form
\begin{eqnarray} \fl
    \quad \mathcal{E}_j^{\prime} =
    \{ E_{M(E_*;\alpha_0)+(j-1)L + \sum_{k=1}^{j-1} n_{k} +1 }^{(\alpha_0)} -
        E_{M(E_*;\alpha_0)+(j-1)L + \sum_{k=1}^{j-1} n_{k} }^{(\alpha_0)}, \dots, \nonumber\\
    \fl \qquad \qquad \quad
    E_{M(E_*;\alpha_0)+(j-1)L + \sum_{k=1}^{j-1} n_{k} + L }^{(\alpha_0)} -
        E_{M(E_*;\alpha_0)+(j-1)L + \sum_{k=1}^{j-1} n_{k} }^{(\alpha_0)}\}.
\end{eqnarray}
The averaging procedure is then performed over the ensemble ${\mathcal E}^{\prime} = \{ {\mathcal E}_1^{\prime},\dots, {\mathcal E}_Q^{\prime}\}$ of $Q \gg 1$ eigenlevel sequences, each of length $L$.
\hfill $\blacksquare$
\newline\newline\noindent
{\bf Protocol 2.}~Contrary to Protocol 1, in which an ensemble of $Q$ spectral sequences was obtained through random sampling of a single, deterministic and sufficiently long sequence of eigenlevels of an {\it individual} billiard, the Protocol 2 deals with an ensemble ${\mathcal E}^{\prime\prime} = \{ {\mathcal E}_1^{\prime\prime},\dots, {\mathcal E}_Q^{\prime\prime}\}$ of eigenlevel sequences originating from $Q$ rectangular billiards whose aspect ratios $\{ \alpha_1,\dots,\alpha_Q\}$ comprise independent random variables. The sides of the $j$-th billiard equal $a=2\sqrt{\pi} \alpha_j^{1/4}$ and $b=2\sqrt{\pi} \alpha_j^{-1/4}$ such that a billiard area ${\mathcal A}=4\pi$ does not vary throughout the ensemble. For each $\alpha_j$, we produce a sufficiently long, ordered, eigenlevel sequence $\{E_1^{(\alpha_j)} \le E_2^{(\alpha_j)} \le \dots \}$ by virtue of Eq.~(\ref{Enm-alpha}). Next, we skip a sufficiently large number $M>10^9$ of lowest eigenlevels to obtain a truncated sequence of $L$ ordered eigenlevels
$\{E_{M+1}^{(\alpha_j)} < E_{M+2}^{(\alpha_j)} < \dots <  E_{M+L}^{(\alpha_j)}\}$ taking particular care that there are no spectral degeneracies therein, within the machine precision. Finally, we tune this subsequence, as discussed in Section~\ref{PS-trunc-ss} for the power spectrum (Definition~\ref{def-01-truncated} and Remark~\ref{rem-collective}) and in Section~\ref{SFF-trunc} for the form factor (Definition~\ref{def-01-K-tau-truncated}, Lemma~\ref{PS-stationary-K-truncated} and Remark~\ref{FF-rem}) to obtain the $j$-th member of the spectral ensemble,
\begin{eqnarray}
    \mathcal{E}_j^{\prime\prime} = \{E_{M+1}^{(\alpha_j)} - E_{M}^{(\alpha_j)},  E_{M+2}^{(\alpha_j)} - E_{M}^{(\alpha_j)}, \dots,  E_{M+L}^{(\alpha_j)} - E_{M}^{(\alpha_j)}\}.
\end{eqnarray}
The averaging procedure is then performed over the ensemble ${\mathcal E}^{\prime\prime} = \{ {\mathcal E}_1^{\prime\prime},\dots, {\mathcal E}_Q^{\prime\prime}\}$ of $Q \gg 1$ eigenlevel sequences, each of length $L$.
\hfill $\blacksquare$
\newline\newline\noindent
{\bf Protocol 3.}~There exists an alternative way to get rid of collective fluctuations in the truncated sequence of eigenlevels, as described below. To build an appropriate ensemble ${\mathcal E}^{\prime\prime\prime} = \{ {\mathcal E}_1^{\prime\prime\prime},\dots, {\mathcal E}_Q^{\prime\prime\prime}\}$ of $Q$ eigenlevel sequences, we again generate $Q$ aspect ratios $\{ \alpha_1,\dots,\alpha_Q\}$ at random, see Protocol 2. For each $\alpha_j$, we compute a sufficiently long, ordered, eigenlevel sequence $\{E_1^{(\alpha_j)} \le E_2^{(\alpha_j)} \le \dots \}$ by virtue of Eq.~(\ref{Enm-alpha}). However, in distinction to Protocol 2, a truncated sequence of $L$ ordered eigenlevels will be constructed in a slightly different manner. Specifically, we choose a high-lying threshold energy $E_*>10^9$ (measured in units of the mean level spacing $\Delta=1$) and count the number of eigenvalues {\it below} the threshold $E_*$; say, there will be $M(E_*;\alpha_j)$ of such eigenvalues; all of them will be discarded. The number $M(E_*;\alpha_j)$ is a random variable as is the aspect ratio $\alpha_j$. The associated truncated sequence of $L$ ordered eigenlevels comprising the $j$-th member of the spectral ensemble reads:
\begin{eqnarray}
    \mathcal{E}_j^{\prime\prime\prime} =
    \{ E_{M(E_*;\alpha_j)+1}^{(\alpha_j)}, E_{M(E_*;\alpha_j)+2}^{(\alpha_j)}, \dots, E_{M(E_*;\alpha_j)+L}^{(\alpha_j)}\}.
\end{eqnarray}
We verify, within the machine precision, that $\mathcal{E}_j^{\prime\prime\prime}$ is free of spectral degeneracies. It should be noticed that spectral subsequence $\mathcal{E}_j^{\prime\prime\prime}$ does not exhibit collective fluctuations since its lowest eigenvalue -- the first above the threshold energy -- may only fluctuate on the scale of the mean level spacing. The averaging procedure is then performed over the ensemble ${\mathcal E}^{\prime\prime\prime} = \{ {\mathcal E}_1^{\prime\prime\prime},\dots, {\mathcal E}_Q^{\prime\prime\prime}\}$ of $Q \gg 1$ eigenlevel sequences, each of length $L$.
\hfill $\blacksquare$
\newline\newline\noindent
The Protocol 3 was applied to study the power spectrum and the spectral form factor in an ensemble of $Q=10^6$ rectangular billiards with randomly chosen aspect ratios $\alpha_j\in (1,2)$ at energies larger than the threshold energy $E_* =10^{13}$. For each billiard realization, a set of $L=2048$ ordered eigenlevels was recorded for further statistical processing. Remarkably, the power spectrum~\footnote[6]{The power spectrum curve for rectangular billiards has been displayed in Fig.~\ref{Plot_1} of the Introduction.} and the form factor computed numerically for an ensemble of rectangular billiards followed closely the theoretical curves Eqs.~(\ref{SL-omega}) and (\ref{KL-tau}), respectively. In fact, a correspondence between simulations and theoretical predictions for irrational rectangular billiards is even better than for semicircular billiards. This can be attributed to a higher threshold energy $E_*$ in the former case. Hence, the discussion and conclusions presented in the end of Section~\ref{SCBN} equally apply to rectangular billiards. Numerical implementation of Protocols 1 and 2 produced very similar curves for the two statistics, thus validating all three computational protocols.

\section*{Acknowledgments}
This work was supported by the Israel Science Foundation through the Grants No.~648/18 (E.K. and R.R.) and No.~2040/17 (R.R.). Some of the computations presented in this work were performed on the Hive computer cluster at the University of Haifa, which is partially funded through the ISF grant No.~2155/15. Last but not least, we thank Dr.~Dmitrii Karp for useful discussion of the properties of hypergeometric functions.

\newpage
\renewcommand{\appendixpagename}{\normalsize{Appendices}}
\addappheadtotoc
\appendixpage
\renewcommand{\thesection}{\Alph{section}}
\renewcommand{\theequation}{\thesection.\arabic{equation}}
\setcounter{section}{0}
\fancyhead{} \fancyhead[RE,LO]{\textsl{\thesection \quad Asymptotics of the Kummer confluent hypergeometric function ${}_1 F_1 (a,b;i x)$}}
\fancyhead[LE,RO]{\thepage}

\section{Asymptotics of the Kummer confluent hypergeometric function ${}_1 F_1 (a,b;i x)$}\label{A-1}

For any $x\in \CC$ and ${\rm Re\,} b > {\rm Re\,} a>0$, the Kummer confluent hypergeometric function ${}_1 F_1 (a,b;i x)$ admits the integral representation
\begin{equation}\label{eq:1F1int}
{}_1F_1(a;b;i x)=\frac{\Gamma(b)}{\Gamma(a)\Gamma(b-a)} I(a;b;x),
\end{equation}
where
\begin{equation}\label{eq:Iabz}
    I(a;b;x)= \int_0^1 e^{i x t} t^{a-1} (1-t)^{b-a-1} dt.
\end{equation}

\begin{lemma}\label{lemma:bound}
For all $b\ge a\ge 0$ and $x \in \RR$, the following bound holds:
\begin{equation}
\left| {}_1F_1	(a; b; i x)\right|\le 1.
\end{equation}	
\end{lemma}
\begin{proof}
For any $x\in\RR$ and $0<a<b$ we have
\begin{eqnarray} \fl
\qquad \left| {}_1F_1	(a; b; i x)\right| = \frac{\Gamma(b)}{\Gamma(a)\Gamma(b-a)}\left| I(a;b;x) \right| \nonumber\\
    \le \frac{\Gamma(b)}{\Gamma(a)\Gamma(b-a)} \int_0^1 \left| e^{i x t} t^{a-1} (1-t)^{b-a-1}\right| dt \nonumber\\
    = \frac{\Gamma(b)}{\Gamma(a)\Gamma(b-a)} \int_0^1  t^{a-1} (1-t)^{b-a-1} dt= {}_1F_1	(a; b; 0)=1.
\end{eqnarray}
Since ${}_1F_1(0; b; i x)=1$ and ${}_1F_1(a; a; i x)=e^{i x}$, the bound also holds for $a=0$ and $a=b$.
\end{proof}

\begin{proposition}\label{prop1new2}
	   Let $a>0$, $x\in \RR$ and $j\in\RR$. Then, as $\lambda \rightarrow \infty$, the expansion
    \begin{eqnarray} \label{eq:limit1new}
	 {}_1 F_1 (a; \lambda + j; i \lambda x) = (1-i x)^{-a} + {\mathcal O}\left( \frac{1}{\lambda} \right)
	\end{eqnarray}
	holds uniformly for any bounded $a$, $x$ and $j$.
\end{proposition}

\begin{proof}
As soon as ${}_1F_1	(a; b; i 0)=1$, we shall only consider $x \neq 0$. Furthermore, it is sufficient to treat the case $x> 0$ only since the domains $x>0$ and $x<0$ are related to each other by virtue of the symmetry relation $I(a;b;-x)=I^*(a;b;x)$.

Owing to the integral representation Eq.~(\ref{eq:1F1int}), we shall focus on the asymptotic analysis of $I(a,\lambda+j,\lambda x)$, Eq.~(\ref{eq:Iabz}), as $\lambda \rightarrow \infty$. This will be done by applying the Laplace method~\cite{Wong} after deforming the integration contour $t\in (0,1)$ in the integral
\begin{eqnarray}\label{eq:int1new}
I(a,\lambda+j,\lambda x) =\int_{0}^1 t^{a-1} (1-t)^{j-a-1} e^{-\lambda h_0(x;t)} dt,\nonumber\\
\label{eq:halpha}
    \qquad \quad h_{0}(x;t) =-i x t - \log(1-t),
\end{eqnarray}
into the integration path $\gamma$ lying in the complex plane
\begin{equation}\label{eq:int1new-gamma}
I(a,\lambda+j,\lambda x)=\int_{\gamma} z^{a-1} (1-z)^{j-a-1} e^{-\lambda h_0(x;z)} dz
\end{equation}
and coinciding, at least partially, with a path $\gamma_0$ of steepest descent determined by equation
\begin{eqnarray}\label{steep-d-eq}
    {\rm Im\,} h_0(x;z)\big|_{z\in \gamma_0} = {\rm Im\,} h_0(x;z=0) = 0.
\end{eqnarray}
The latter brings an explicit parameterization of the path $\gamma_0$ in the form $z= 1-r e^{-i\phi}$, where
\begin{eqnarray}
    \phi=x(1-r\cos\phi).
\end{eqnarray}
In particular, if $x=0$, the real line will be a path of steepest descent.

For $x>0$, a path of steepest descent is described by the curve
\begin{eqnarray}\label{cont-g-0}
    \gamma_0: \;\; z=z_0(\phi) = \frac{\phi}{x} + i \left( 1 - \frac{\phi}{x}\right) \tan \phi, \quad \phi\in(0,\phi_0),
\end{eqnarray}
where $\phi_0 = {\rm min}(x,\pi/2)$. Depending on the values of $x$, Eq.~(\ref{cont-g-0}) brings out three different types of the deformed integration contour $\gamma$.

(i) If $0<x<\pi/2$, the path of steepest descent $\gamma_0$ is a finite-length curve starting at the origin $z=0$ and finishing at the endpoint $z=1$, see the left panel in Fig.~\ref{fig:SP1}. Hence, the integration path $\gamma$ in Eq.~(\ref{eq:int1new-gamma}) is identical to $\gamma_0$.

(ii) If $x = \pi/2$, the path of steepest descent $\gamma_0$ is a finite-length curve starting at the origin $z=0$ and finishing at the point $z= 1+2i/\pi$ thus {\it not} reaching the endpoint $z=1$. To close the integration contour, we need to add the straight (steepest descent) path $\gamma_1$ connecting the points $z= 1+2i/\pi$ and $z=1$, see the middle panel in
Fig.~\ref{fig:SP1}. Hence, the integration path $\gamma$ in Eq.~(\ref{eq:int1new-gamma}) is $\gamma=\gamma_0 \cup \gamma_1$.

(iii) Finally, if $x> \pi/2$, the path of steepest descent $\gamma_0$ is a curve of infinite length starting at the origin $z=0$ and approaching complex infinity at $z=\pi/(2x) + i\infty$. To build the integration contour $\gamma$, we need to return from complex infinity along the straight (steepest descent) path $\gamma_2$ connecting the points $z= 1+i\infty$ and $z=1$, and connecting the curves $\gamma_0$ and $\gamma_2$ with an arc $\gamma_\cap$, see the right panel in Fig.~\ref{fig:SP1}. Hence, the integration path $\gamma$ in Eq.~(\ref{eq:int1new-gamma}) is $\gamma=\gamma_0 \cup \gamma_\cap \cup \gamma_2$.
\begin{figure}[b]
\begin{center}
\includegraphics[width=0.33\textwidth]{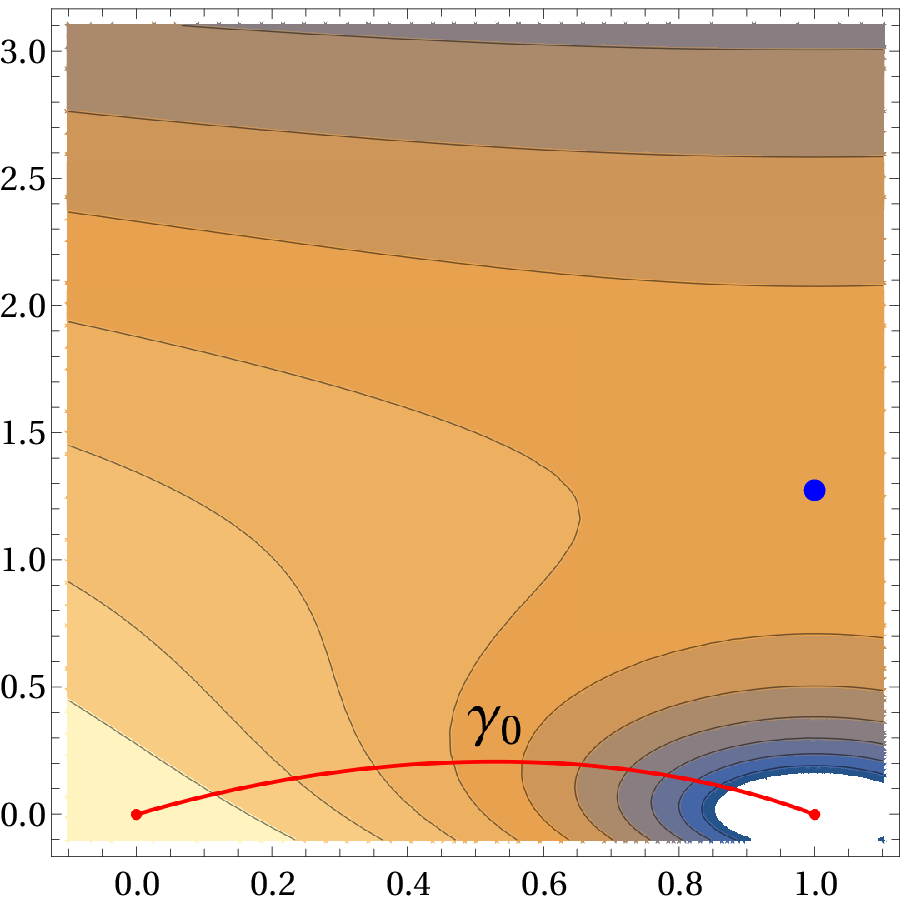}\hfill\includegraphics[width=0.33\textwidth]{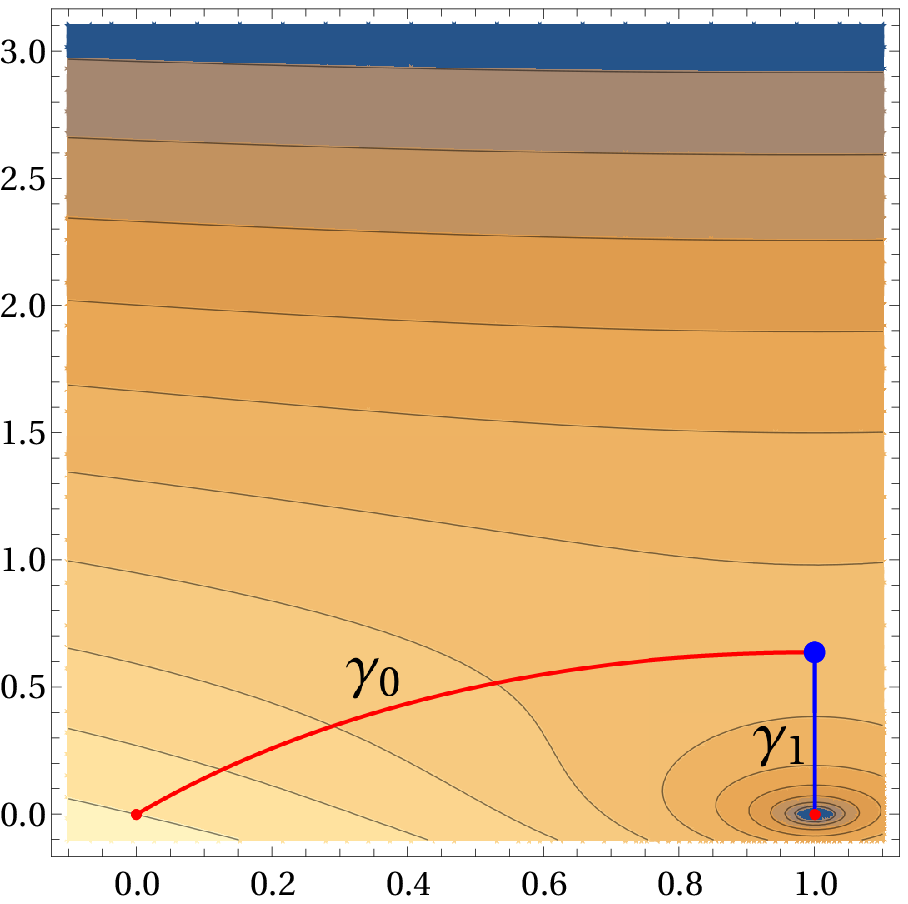}\hfill\includegraphics[width=0.33\textwidth]{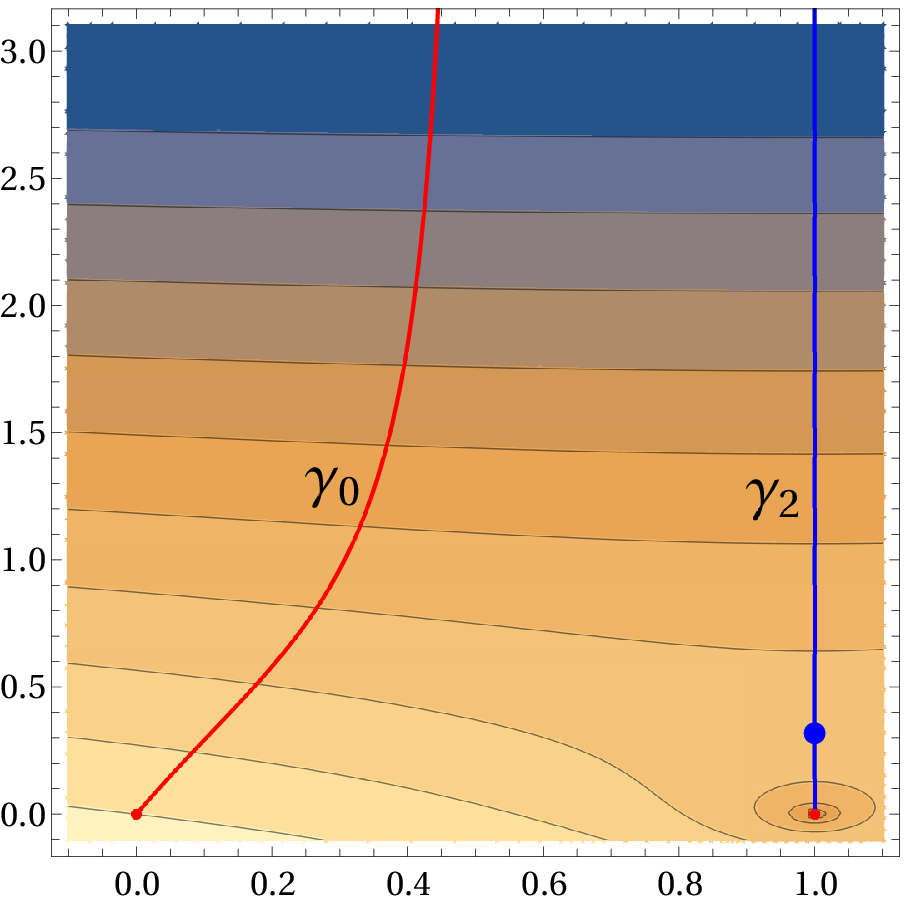}
\caption{Integration path $\gamma$ in the integral Eq.~(\ref{eq:int1new-gamma}) for three cases. Left panel: $0 < x< \pi/2$. Middle panel: $x=\pi/2$. Right panel: $x>\pi/2$.
Colored background is a contour plot of ${\rm Re\,} h_0(x; z)$ as a function of $z\in\CC$. Blue dot denotes the saddle point of $h_0(x;z)$.}\label{fig:SP1}
\end{center}
\end{figure}

Notice that for $x \ge \pi/2$  the integration over additional paths $\gamma_1$ and $\gamma_2$ will bring exponentially small corrections to the dominating contribution coming from the contour $\gamma_0$ as $\lambda\rightarrow \infty$. This holds true since for any $z$ belonging to $\gamma_1,\gamma_2$ we observe the inequality ${\rm Re\,}h_0(x;z) > h_0(x;0)=0$. The integration over $\gamma_\cap$ brings no contribution even for finite $\lambda$.

Hence, we are left to consider a contribution $I_{\gamma_0}(a,\lambda+j,\lambda x)$ to the integral Eq.~(\ref{eq:int1new-gamma}) coming from a path of steepest descent $\gamma_0$. Making use of Eq.~(\ref{cont-g-0}), we obtain:
\begin{eqnarray}\label{eq:int1new-gamma-000}
\fl
I_{\gamma_0}(a,\lambda+j,\lambda x)=\int_{\gamma_0} z^{a-1} (1-z)^{j-a-1} e^{-\lambda h_0(x;z)} dz = \int_0^{\phi_0} d\phi \, f(\phi) \, e^{-\lambda h(\phi)},
\end{eqnarray}
where
\begin{eqnarray} \fl
\qquad f(\phi) = z_0^\prime(\phi) \, z_0(\phi)^{a-1} (1-z_0(\phi))^{j-a-1} =\left(1 + i\frac{x-\phi}{\cos^2\phi} - i \tan\phi\right) \nonumber\\
 \times  x^{1-j} (x-\phi)^{j-a-1} \left(1 - i \tan \phi\right)^{j-a-1}
    \left( \phi + i (x-\phi)\tan \phi \right)^{a-1}
\end{eqnarray}
and
\begin{eqnarray}\fl \label{hf}
\qquad h(\phi) = h_0(x; z_0(\phi)) = (x-\phi) \tan \phi -\log\left(1 -\frac{\phi}{x} \right) + \log \cos\phi.
\end{eqnarray}
Since the function $h(\phi)$, appearing in the exponent of the integral Eq.~(\ref{eq:int1new-gamma-000}), reaches its minimum at the origin $\phi=0$ and $h(0)=0$, the main contribution to the integral $I_{\gamma_0}(a,\lambda+j,\lambda x)$ comes from a neighborhood of $\phi=0$. To determine the dominating contribution to $I_{\gamma_0}(a,\lambda+j,\lambda x)$, we make use of Theorem 1 from Chapter II in Ref.~\cite{Wong} to write down a $\lambda\rightarrow \infty$ asymptotic expansion
\begin{eqnarray} \label{as-exp}
    I_{\gamma_0}(a,\lambda+j,\lambda x) = e^{-\lambda h(0)}\sum_{k=0}^\infty \Gamma(k+a)\frac{c_k}{\lambda^{k+a}},
\end{eqnarray}
where the coefficients $c_k$ can be related to the coefficients $a_k$ and $b_k$ arising in the expansion of the functions $h$ and $f$ at $\phi=0$,
\begin{eqnarray}\label{eq:hfexp}
    h(\phi)=h(0)+\sum_{k=0}^\infty a_k \phi^{k+1}, \qquad f(\phi)=\sum_{k=0}^\infty b_k \phi^{k+a-1}.
\end{eqnarray}
The leading order term in Eq.~(\ref{as-exp}) originates from $k=0$ for which $c_0=b_0/a_0^a$. Given that
\begin{eqnarray}\label{eq:ab0}
    a_0=x+\frac{1}{x}, \qquad b_0=\left(i+\frac{1}{x}\right)^a,
\end{eqnarray}
we eventually derive the asymptotic expansion of $I_{\gamma_0}(a,\lambda+j,\lambda x)$ in the form
\begin{eqnarray}\label{eq:resIphi}
I_{\gamma_0}(a,\lambda+j,\lambda x) = \frac{\Gamma(a)}{(1- i x)^a \lambda^a}+\mathcal{O}(\lambda^{-a-1})
\end{eqnarray}
as $\lambda\rightarrow\infty$. This essentially finishes the proof after noticing that the pre-factor in Eq.~(\ref{eq:1F1int}) exhibits an asymptotic
behavior
\begin{equation} \label{G-as}
\frac{\Gamma(\lambda+j)}{\Gamma(\lambda+j-a)}=\lambda^a+\mathcal{O}(\lambda^{a-1})
\end{equation}
as $\lambda\rightarrow\infty$ while $a$ and $j$ are kept fixed.

Finally, since for any $x\in\RR$ the saddle point $\phi_*$ of $h_0$, occuring at $\phi_*=1+i/x$, stays away from the origin $\phi=0$ at which the global minimum of $h_0$ is achieved along the path $\gamma$, we conclude that the asymptotic expansion holds uniformly in $x$. Furthermore, since the integrand is regular in the vicinity of the origin $\phi=0$ for any $j$ and $a>0$, the uniformity in $j$ and $a$ is also secured.

\end{proof}

\begin{proposition}\label{1F1-as}
Let $0<b<1$, $x\in\mathbb{R}$ and $j\in{\mathbb R}$. Then, as $\lambda\rightarrow\infty$, the asymptotic expansion   	
	\begin{eqnarray}\label{eq:fixedT1} \fl
    \qquad {}_1 F_1 &(b \lambda+j; \lambda+j; i x) =  e^{i b x}
	\left(
	1 - \frac{x(1-b)(bx-2ij)}{2\lambda} + {\mathcal O}(\lambda^{-2})
	\right)
	\end{eqnarray}
holds uniformly for any bounded $x$ and $j$.
\end{proposition}

\begin{proof}
With Eq.~(\ref{eq:1F1int}) in mind, we shall focus on the asymptotic analysis of $I(b\lambda+j,\lambda+j,\lambda x)$, Eq.~(\ref{eq:Iabz}), as $\lambda \rightarrow \infty$. To proceed, we start with the integral representation
\begin{equation}\label{eq:int1}
I(b \lambda+j,\lambda+j,x)=\int_0^1 g(x;t) e^{-\lambda h_b(t)} dt,
\end{equation}
where
\begin{eqnarray}\label{eq:hb}
 h_b(t) &=&  -b \log t -(1-b)\log(1-t)
\end{eqnarray}
and
\begin{eqnarray}
    \label{eq:g}
    g(x;t)&=& \frac{t^{j-1}}{1-t} e^{i x t}.
\end{eqnarray}
Real-valuedness of the exponent $h_b(t)$ allows us to apply the Laplace method, see Theorem 1 from Chapter II in Ref.~\cite{Wong}. A saddle point $t_*$ of
the large-$\lambda$ integral Eq.~(\ref{eq:int1}) is seen to be located at $t_*=b$ at which $h_b(t)$ reaches its minimum for $t\in [0,1]$. In view of Theorem 1 of Ref.~\cite{Wong}, expansions
\begin{eqnarray}\label{eq:hfexp}
h_b(t)=h_b(b)+\sum_{k=0}^\infty \alpha_k (t-b)^{k+2}
\end{eqnarray}
and
\begin{eqnarray}\label{eq:hfexp-2}
g(x;t)=\sum_{k=0}^\infty \beta_k(x) (t-b)^{k}
\end{eqnarray}
written in a vicinity of $t_*=b$ are of particular (technical) interest. There,
\begin{eqnarray} \fl
\qquad \alpha_0=\frac{1}{2b(1-b)},\quad \alpha_1=\frac{2b-1}{3b^2(1-b)^2},\quad \alpha_2=\frac{3b^2-3b+1}{4b^3(1-b)^3},
\end{eqnarray}
and
\begin{eqnarray} \fl
\qquad \beta_0(x)=e^{ibx} \frac{b^{j-1}}{1-b},\quad \beta_1(x)=e^{i b x}\frac{b^{j-2}}{(1-b)^2} \left( 2b-1 +(1-b)(ibx+j)\right),\nonumber\\
\fl \qquad
\beta_2(x)= e^{ibx} \frac{b^{j-3}}{2(b-1)^3}\bigg\{b^2x^2(1-b)^2 + 2i bx\left[b^2(2-j)-b(3-j) -j+1\right]\nonumber\\
\fl \qquad\qquad\qquad\quad -2b^2(2-j)(3-j)+2b(3-j)(1-j)-(2-j)(1-j) \bigg\}.
\end{eqnarray}
To meet requirements of the aforementioned Theorem 1, the point $t_*$ at which the exponent $h_b(t)$ reaches its minimum should coincide with a lower limit of the integral Eq.~(\ref{eq:int1}). To achieve this, one has to split the integral at the saddle point $t_*$ into two integrals over domains $(0, t_*)$ and $(t_*,1)$. It appears that their asymptotic expansions, as $\lambda\rightarrow\infty$, partially cancel each other furnishing a series
\begin{eqnarray}\label{eq:int2}
I(b\lambda+j,\lambda+j,x)=2 e^{-\lambda h_b(b)}\sum_{k=0}^\infty \Gamma\left(\frac{2k+1}{2}\right)\frac{\gamma_{2k}(x)}{\lambda^{(2k+1)/2}}
\end{eqnarray}
containing functions $\gamma_{2k}(x)$ labeled with even indices. Here,
\begin{eqnarray}\fl \quad
 \gamma_0(x) =\frac{1}{2\alpha_0^{1/2}}\,\beta_0(x)=e^{i bx}\frac{b^j}{\sqrt{2b(1-b)}},\nonumber\\
 \fl \quad \gamma_2(x)=\frac{1}{2\alpha_0^{3/2}}\left( \beta_2(x)-\frac{3\alpha_1}{2 \alpha_0} \beta_1(x) + \frac{3 (5\alpha_1^2-4\alpha_0 \alpha_2)}{8 \alpha_0^2}\,
 \beta_0(x)\right)\nonumber\\
       \fl \quad  \qquad =e^{i b x}\frac{b^{j-3/2}}{6\sqrt{2} (1-b)^{3/2}} \, \left[ (1-b)^2(-6b^2x^2+12ibjx+6j^2-6j+1) + b\right].
\end{eqnarray}
Keeping the $k=0$ and $k=1$ contributions in Eq.~(\ref{eq:int2}), we derive
\begin{eqnarray}\label{eq:Iresult}\fl
I(b \lambda+j,\lambda+j,x)= e^{i bx} \frac{\sqrt{2\pi} b^{\lambda b+j-1/2} (1-b)^{\lambda(1-b)-1/2}}{\sqrt{\lambda}}\nonumber\\
\fl
\qquad\qquad\times \left[1+\frac{(1-b)^2 \left(-6b^2x^2+12ibjx+6j^2-6j+1 \right)+b}{12b(1-b) \lambda} + \mathcal{O}(\lambda^{-2})\right].
\end{eqnarray}
Combined with the large-$\lambda$ expansion of the pre-factor in Eq.~(\ref{eq:1F1int}),
\begin{eqnarray}\label{eq:prefactor1}\fl
\frac{\Gamma(\lambda+j)}{\Gamma(b\lambda+j) \Gamma(\lambda(1-b))} =\frac{\sqrt{\lambda}}{\sqrt{2\pi}  b^{\lambda b+j-1/2} (1-b)^{\lambda(1-b)-1/2}}\nonumber\\
\qquad\times\left[1-\frac{(1-b)^2\left( 6j^2-6j+1 \right)+b}{12b(1-b)\lambda} +\mathcal{O}(\lambda^{-2})\right].
\end{eqnarray}
this yields the sought Eq.~(\ref{eq:fixedT1}).

Finally, we notice that since the minimum of $h_b(t)$ at $t=b$ is independent of $x$ and $j$ and the integrand does not posses singularities in the vicinity of $t=b$, the asymptotic expansion Eq.~(\ref{eq:fixedT1}) holds uniformly for any bounded $x$ and $j$.
\end{proof}

\begin{lemma} \label{L-sp}
Let
\begin{eqnarray} \label{hbxz}
h_b(x;z) = - i x z - b \log z - (1-b) \log (1-z)
\end{eqnarray}
be a function of $z \in {\mathbb C}$ with $0<b<1$ and $x \ge 0$ possessing two saddle points
\begin{eqnarray}
    z_{\pm} = \frac{1}{2} + \frac{i}{2x} \left(
        1 \pm \sqrt{1-2 i x (1-2b) - x^2}
    \right),
\end{eqnarray}
and let $\gamma_b$ be a path in the complex plane such that ${\rm Im} h_b(x;z) = {\rm Im} h_b(x;z_-)$ along $\gamma_b$. Then, for $x$ small enough, such a path connects the points $z=0$ and $z=1$ passing through the saddle point $z_-$ but not through $z_+$.
\end{lemma}
\begin{proof}
    Due to a forthcoming application of this lemma to an asymptotic evaluation of the integral Eq.~(\ref{pst-03}), we shall call $\gamma_b$ a path of steepest descent throughout the proof. Let us parameterize the path $\gamma_b$ passing through the saddle point $z_-$ as
\begin{eqnarray}
    z_b(\phi) = \phi + i \chi_b(\phi),
\end{eqnarray}
where $\chi_b(\phi) \in {\mathbb R}$ is a solution to the equation ${\rm Im}\, h_b(x;z_b(\phi)) = {\rm Im}\, h_b(x;z_-)$ whose explicit form reads
\begin{eqnarray} \fl \label{sd-b}
     - x\phi - b \arctan\left( \frac{\chi_b(\phi)}{\phi}\right)
     - (1-b) \arctan\left( \frac{\chi_b(\phi)}{\phi-1}\right) = {\rm Im}\, h_b(x;z_-).
\end{eqnarray}
Since ${\rm Im}\, h_b(x;z_-)$ is a continuous function of $x$ in a vicinity of $x=0$, then so is its solution $\chi_b(x;\phi)$. As $x\rightarrow 0$, the path $\gamma_b$ approaches the real line, becoming a segment $(0,1)$ at $x=0$. This is consistent with the behavior of $z_-$ as a function of $x$ as $x\rightarrow 0$ such that $z_{-}= b$ at $x=0$.

To prove that $z=0$ belongs to the path of steepest descent for sufficiently small $x$, we must show that $\chi(\phi) \rightarrow 0$ as $\phi \rightarrow 0$. Let us assume that this is not the case, that is  $\chi(\phi) \rightarrow \chi_* \neq 0$. Then the $\phi \rightarrow 0$ limit of Eq.~(\ref{sd-b}) yields
\begin{eqnarray}
    \chi_* = \tan \left[
        \frac{1}{1-b} \left(
            \frac{\pi b}{2} + {\rm Im}\, h_b(x;z_-)
        \right)
    \right]
\end{eqnarray}
which should hold in a vicinity of $x=0$. In particular, at $x=0$, it brings $\chi_* \neq 0$ which contradicts our observation that $\gamma_b$ becomes a segment $(0,1)$ along the real line at $x=0$.  (Essentially, this is secured by continuity of $\chi_b(x;\phi)$.) Hence, $z=0$ does belong to $\gamma_b$. A similar reasoning can be applied to show that the point $z=1$ belongs to $\gamma_b$ as well.

Finally, spotting that $|z_+| \rightarrow \infty$ as $x \rightarrow 0$, we conclude that $z_+$ stays away from $\gamma_b$. The path of steepest descent $\gamma_b$ is illustrated in the left panel of Fig.~\ref{fig:SPPP}.
\end{proof}

\begin{remark}\label{x0-3-cases}
    The topology of the steepest descent path depends on the value of $x$. For $x$ small enough, this path was discussed in Lemma~\ref{L-sp}. As $x$ grows, one approaches a critical situation in which a path of steepest descent goes through both saddle points $z_-$ and $z_+$ as illustrated in the middle panel of Fig.~\ref{fig:SPPP}. A critical value of $x$ is determined by the equation ${\rm Im\,} h_b(x;z_-) = {\rm Im\,} h_b(x;z_+)$. Further increase of $x$ drives a steepest descent path connecting the points $z=0$ and $z_-$ to infinity (see red curve in the right panel of Fig.~\ref{fig:SPPP}). To reach the endpoint $z=1$, one has to choose another steepest descent path originating at infinity and approaching $z=1$ through $z_+$ as depicted by a blue curve in the right panel of Fig.~\ref{fig:SPPP}.
\end{remark}
\begin{figure}[b]
\begin{center}
\includegraphics[width=0.33\textwidth]{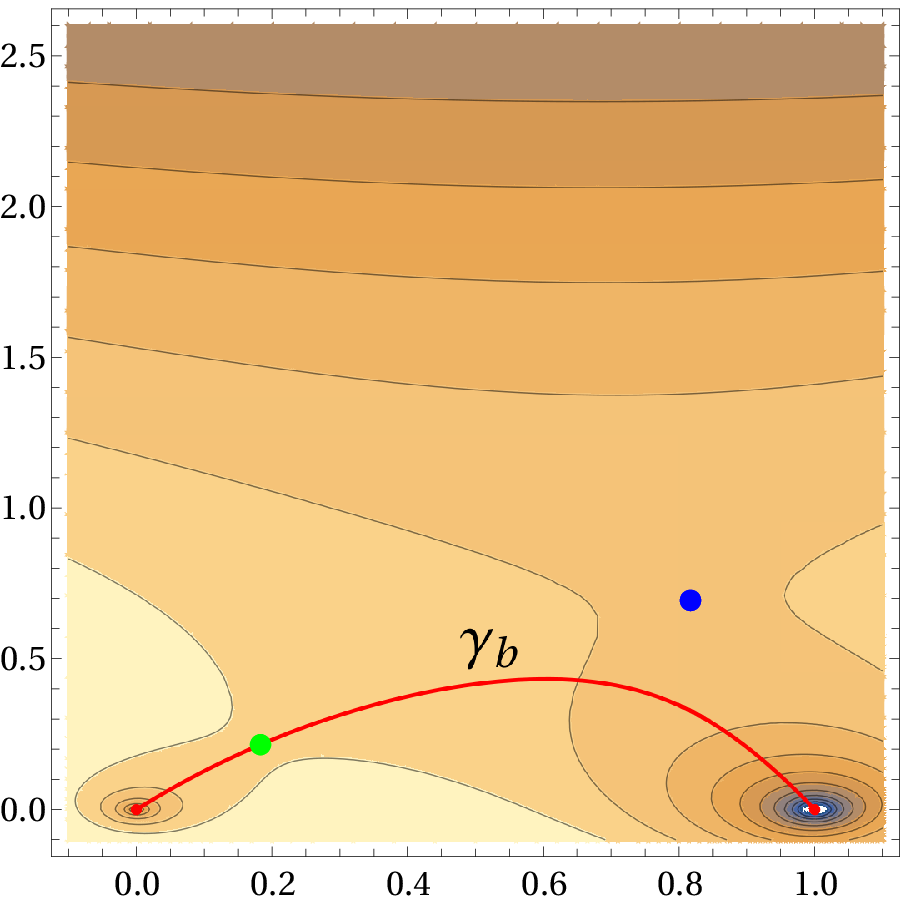}\hfill\includegraphics[width=0.33\textwidth]{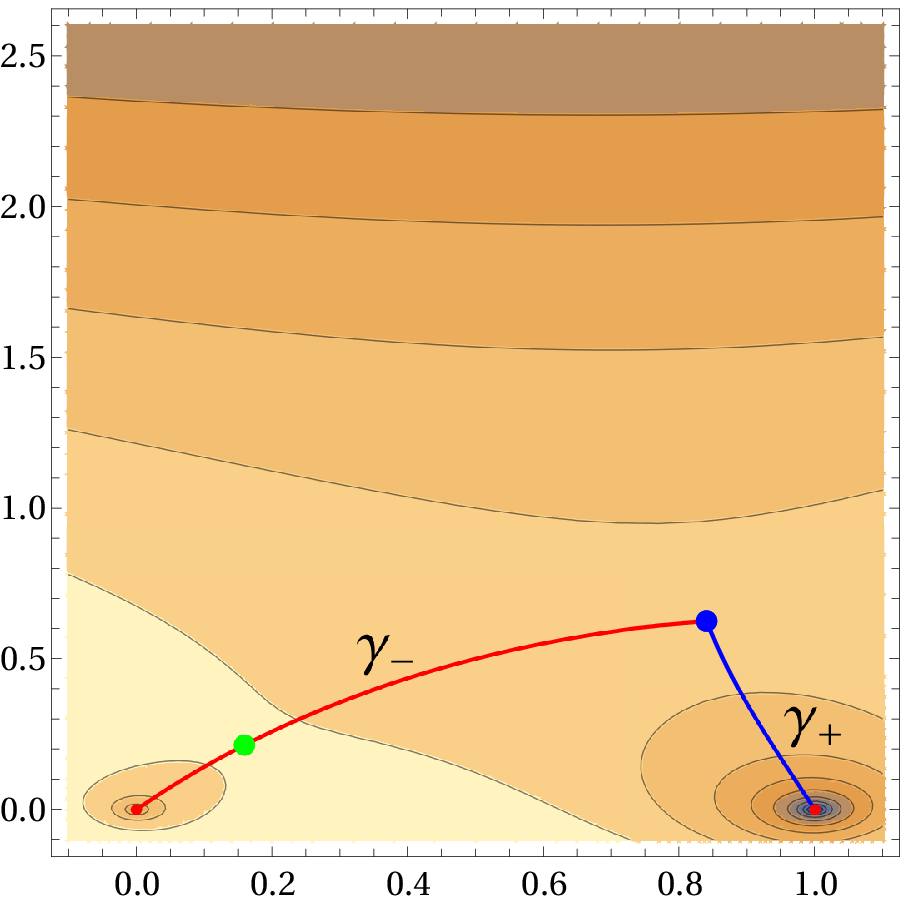}\hfill\includegraphics[width=0.33\textwidth]{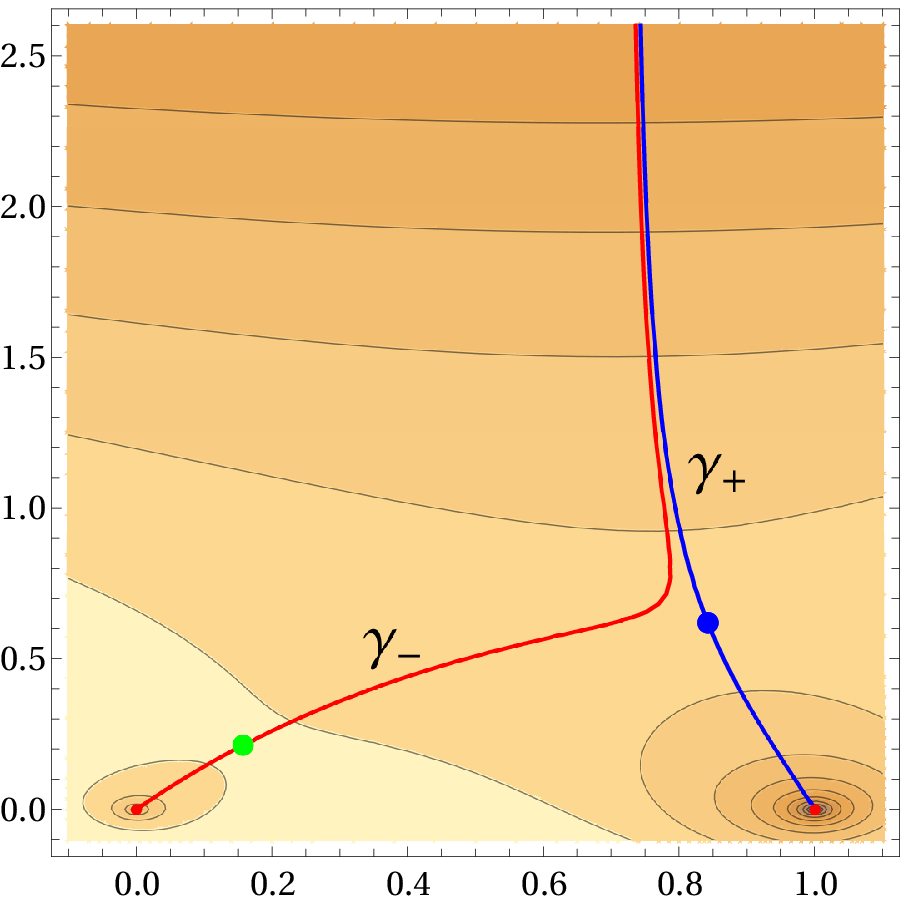}
\caption{Three topologies of a steepest descent path for $h_b(x;z)$ defined by Eq.~(\ref{hbxz}). Left panel: a path $\gamma_b$ described in Lemma~\ref{L-sp} ($x$ is small enough). Middle and right panels illustrate steepest descent paths discussed in Remark~\ref{x0-3-cases}. Colored background is a contour plot of ${\rm Re\,} h_b(x;z)$ as a function of $z\in\CC$. Green and blue dots denote the saddle points $z_-$ and $z_+$, respectively.} \label{fig:SPPP}
\end{center}
\end{figure}

\begin{proposition}\label{prop:1F1new}
Let $0<b<1$, $x\in\mathbb{R}$, $j\in{\mathbb R}$ and $x_0$ be sufficiently small. Then, as $\lambda\rightarrow\infty$, the asymptotic expansion   	
\begin{eqnarray} \fl \label{pst-prop-main}
\qquad {}_1F_1(b \lambda+j;\lambda+j; i x \lambda)= e^{i x z_- \lambda} e^{-\lambda \left( h_b(0;z_-) -h_b(0; b)  \right)} \nonumber\\
    \qquad \qquad \times
\left(
    \left( \frac{z_-}{b} \right)^j \frac{\sqrt{b(1-b)}}{\sqrt{z_-^2 -2 b z_- +b}} + {\mathcal O}\left( \frac{1}{\lambda} \right)
\right),
\end{eqnarray}
where
\begin{eqnarray}\label{z0}
    z_- =  \frac{1}{2} + \frac{i}{2x} \left(
        1 - \sqrt{1-2 i x (1-2b) - x^2}
    \right)
\end{eqnarray}
and
\begin{eqnarray}
    h_b(0;z) = - b \log z - (1-b) \log (1-z),
\end{eqnarray}
holds uniformly for all $x \in [0,x_0]$.
\end{proposition}

\begin{proof}
    Owing to Eq.~(\ref{eq:1F1int}), we have
    \begin{eqnarray}\fl \label{pst-01}
      \qquad  {}_1 F_1 (b \lambda+j;\lambda+j; i x \lambda)= \frac{\Gamma(\lambda+j)}{\Gamma(b(\lambda+j)) \Gamma((1-b)\lambda)}
\, I(b \lambda+j;\lambda+j; x \lambda),
    \end{eqnarray}
where
\begin{eqnarray}\label{pst-02}
        I(b \lambda+j;\lambda+j; x \lambda) = \int_{0}^{1} g_j(t) \, e^{-\lambda h_b(x;t)} dt,
\end{eqnarray}
\begin{eqnarray}
        g_j(t) = \frac{t^{j-1}}{1-t},
\end{eqnarray}
and $h_b(x;t)$ is defined in Eq.~(\ref{hbxz}). Since the asymptotic behavior of gamma functions in Eq.~(\ref{pst-01}), as $\lambda \rightarrow \infty$, is well known,
\begin{eqnarray} \label{G-3-exp} \fl
    \qquad \frac{\Gamma(\lambda+j)}{\Gamma(b(\lambda+j)) \Gamma((1-b)\lambda)} = \sqrt{{\frac{b(1-b)}{2\pi}}} \frac{\sqrt{\lambda}}{b^j}
     \, e^{\lambda h_b(0;b)} \left( 1 + {\mathcal O}\left( \frac{1}{\lambda} \right) \right),
\end{eqnarray}
we turn to the asymptotic evaluation of the integral Eq.~(\ref{pst-02}).

To proceed, we deform the integration path therein to a path of steepest descent, $\gamma_b$, lying in the complex plane and connecting the endpoints $z=0$ and $z=1$ of the integral through the saddle point $z_-$, see Eq.~(\ref{z0}),
\begin{eqnarray}\label{pst-03}
        I(b \lambda+j;\lambda+j; x \lambda) = \int_{\gamma_b} g_j(z) \, e^{-\lambda h_b(x;z)} dz.
\end{eqnarray}
Existence of such a path for $x$ small enough was proven in Lemma~\ref{L-sp}. As $\lambda\rightarrow \infty$, the integral Eq.~(\ref{pst-03}) is dominated by a vicinity of the saddle point $z_-$:
\begin{eqnarray}\label{pst-04}\fl
\qquad        I(b \lambda+j;\lambda+j; x \lambda) = \int_{\gamma_b} g_j(z) \, e^{-\lambda h_b(x;z)} dz \nonumber\\
        = \sqrt{\frac{2\pi}{\lambda}} g_j(z_-) e^{-\lambda h_b(x;z_-)} \frac{1}{\sqrt{h_b^{\prime\prime}(x;z_-)}}
         \left( 1 + {\mathcal O}\left( \frac{1}{\lambda} \right) \right).
\end{eqnarray}
Combining Eqs.~(\ref{pst-01}), (\ref{G-3-exp}) and (\ref{pst-04}), we end up with Eq.~(\ref{pst-prop-main}).
\end{proof}

\begin{remark}
  There are strong indications that the asymptotic result Eq.~(\ref{pst-prop-main}) holds true irrespective of the smallness of $x$ provided that $b \neq 1/2$. According to Remark~\ref{x0-3-cases}, an asymptotic evaluation of the integral Eq.~(\ref{pst-03}) for larger values of $x$ necessitates a deformation of the integration contour which has to pass through {\it both} saddle points $z_-$ and $z_+$, see the middle and right panels in Fig.~\ref{fig:SPPP}. Interestingly, even though both saddle points belong to the integration path, numerics provides an evidence that the integral Eq.~(\ref{pst-03}) is still dominated by the contribution of $z_-$. This is not the case for $b=1/2$ when both saddle points become equally important.
\end{remark}

\newpage

\section*{References}
\fancyhead{} \fancyhead[RE,LO]{\textsl{References}}
\fancyhead[LE,RO]{\thepage}


\begin{thebibliography}{9}

\bibitem{B-1987} M. V. Berry:
    Quantum chaology.
    Proc. R. Soc. A {\bf 413}, 183 (1987).

\bibitem{M-2004} M. L. Mehta:
    {\it Random Matrices} (Elsevier, Amsterdam, 2004).

\bibitem{GMGW-1998} T.~Guhr, A.~M\"uller-Groeling, and H.~A.~Weidenm\"uller:
    Random-matrix theories in quantum physics: Common concepts.
    Phys. Reports {\bf 299}, 189 (1998).

\bibitem{S-1999} H.-J.~St\"ockmann:
    {\it Quantum Chaos: An Introduction} (Cambridge University Press, Cambridge, 1999).

\bibitem{H-2001} F.~Haake:
    {\it Quantum Signatures of Chaos} (Springer, Berlin, 2001).

\bibitem{BGS-1984} O.~Bohigas, M.~J.~Giannoni, and C.~Schmit:
    Characterization of chaotic quantum spectra and universality of level fluctuation laws.
    Phys. Rev. Lett. {\bf 52}, 1 (1984).

\bibitem{PF-book} P.~J.~Forrester:
    {\it Log-Gases and Random Matrices} (Princeton University Press, Princeton NJ, 2010).

\bibitem{MK-1979} S.~W.~McDonald and A.~N.~Kaufman:
    Spectrum and eigenfunctions for a Hamiltonian with stochastic trajectories.
    Phys. Rev. Lett. {\bf 42}, 1189 (1979).

\bibitem{CVG-1980} G.~Casati, F.~Valz-Gris, and I.~Guarnieri:
    On the connection between quantization of nonintegrable systems and statistical theory of spectra.
    Lettere Nuovo Cimento {\bf 28}, 279 (1980).

\bibitem{B-1981} M.~V.~Berry:
    Quantizing a classically ergodic system: Sinai's billiard and the KKR method.
    Ann. Phys. {\bf 131}, 163 (1981).

\bibitem{AASA-1996} A.~V.~Andreev, O.~Agam, B.~D.~Simons, and B.~L.~Altshuler:
    Quantum chaos, irreversible classical dynamics, and random matrix theory.
    Phys. Rev. Lett. {\bf 76}, 3947 (1996);
    O.~Agam, A.~V.~Andreev, and B.~D.~Simons:
    Quantum chaos: A field theory approach.
    Chaos, Solitons \& Fractals {\bf 8}, 1099 (1997)

\bibitem{RS-2002} K.~Richter and M.~Sieber:
    Semiclassical theory of chaotic quantum transport.
    Phys. Rev. Lett. {\bf 89}, 206801 (2002);
    S.~M\"uller, S.~Heusler, P.~Braun, F.~Haake, and A.~Altland:
    Semiclassical foundation of universality in quantum chaos.
    Phys. Rev. Lett. {\bf 93}, 014103 (2004);
    S.~Heusler, S.~M\"uller, A.~Altland, P.~Braun, and F.~Haake:
    Periodic-orbit theory of level correlations.
    Phys. Rev. Lett. {\bf 98}, 044103 (2007);
    S.~M\"uller, S.~Heusler, A.~Altland, P.~Braun, and F.~Haake:
    Periodic-orbit theory of universal level correlations in quantum chaos.
    New J. Phys. {\bf 11}, 103025
    (2009).

\bibitem{BT-1977} M.~V.~Berry and M.~Tabor:
    Level clustering in the regular spectrum.
    Proc. R. Soc. A {\bf 356}, 375 (1977).

\bibitem{JMMS-1980} M.~Jimbo, T.~Miwa, Y.~M\^{o}ri, and M.~Sato:
    Density matrix of an impenetrable Bose gas and the fifth Painlev\'{e} transcendent.
    Physica D {\bf 1}, 80 (1980).

\bibitem{ABGR-2013} Y.~Y.~Atas, E.~Bogomolny, O.~Giraud, and G.~Roux:
    Distribution of the ratio of consecutive level spacings in random matrix ensembles.
    Phys. Rev. Lett. {\bf 110}, 084101 (2013).

\bibitem{B-1988} M.~V.~Berry:
    Semiclassical formula for the number variance of the Riemann zeros.
    Nonlinearity {\bf 1}, 399 (1988).

\bibitem{B-1985} M.~V.~Berry:
    Semiclassical theory of spectral rigidity.
    Proc. R. Soc. A {\bf 400}, 229 (1985).

\bibitem{GRRFSVR-2005} J.~M.~G.~G\'{o}mez, A.~Rela\~{n}o, J.~Retamosa, E.~Faleiro, L.~Salasnich, M.~Vrani\v{c}ar, and M.~Robnik:
    $1/f^\alpha$ noise in spectral fluctuations of quantum systems.
    Phys. Rev. Lett. {\bf 94}, 084101 (2005).

\bibitem{FKMMRR-2006} E.~Faleiro, U.~Kuhl, R.~A.~Molina, L.~Mu\~{n}oz, A.~Rela\~{n}o, and J.~Retamosa:
    Power spectrum analysis of experimental Sinai quantum billiards.
    Phys. Lett. A {\bf 358}, 251 (2006).

\bibitem{LBYBS-2018} M.~\L{}awniczak, M.~Bia\l{}ous, V.~Yunko, S.~Bauch, and L.~Sirko:
    Missing-level statistics and analysis of the power spectrum of level fluctuations of three-dimensional chaotic microwave cavities.
    Phys. Rev. E {\bf 98}, 012206 (2018).

\bibitem{Od-1987} A.~M.~Odlyzko:
    On the distribution of spacings between zeros of the zeta function.
    Math. Comput. {\bf 48}, 273 (1987).

\bibitem{RGMRF-2002} A.~Rela\~{n}o, J.~M.~G.~G\'{o}mez, R.~A.~Molina, J.~Retamosa, and E.~Faleiro:
    Quantum chaos and $1/f$ noise.
    Phys. Rev. Lett. {\bf 89}, 244102 (2002).

\bibitem{ROK-2017} R.~Riser, V.~Al.~Osipov, and E.~Kanzieper:
        Power spectrum of long eigenlevel sequences in quantum chaotic systems.
        Phys. Rev. Lett. {\bf 118}, 204101 (2017).

\bibitem{ROK-2020} R.~Riser, V.~Al.~Osipov, and E.~Kanzieper:
        Nonperturbative theory of power spectrum in complex systems.
        Ann. Phys. {\bf 413}, 168065 (2020).

\bibitem{FGMMRR-2004} E.~Faleiro, J.~M.~G.~G\'{o}mez, R.~A.~Molina, L.~Mu\~{n}oz, A.~Rela\~{n}o, and J.~Retamosa:
        Theoretical derivation of $1/f$ noise in quantum chaos.
        Phys. Rev. Lett. {\bf 93}, 244101 (2004).

\bibitem{PPK-2020} A.~Prakash, J.~H.~Pixley, and M.~Kulkarni:
        The universal spectral form factor for many-body localization.
        arXiv:2008.07547.

\bibitem{G-1996} T.~Guhr:
        Transitions toward quantum chaos: with supersymmetry from Poisson to Gauss.
        Ann. Phys. {\bf 250}, 145 (1996).

\bibitem{M-1959} J.~G.~Mauldon:
        A generalization of the beta-distribution.
        Ann. Math. Stat. {\bf 30}, 509 (1959).

\bibitem{PBM-2002} A.~P.~Prudnikov, Yu.~A.~Brychkov, and O.~I.~Marichev: {\it Integrals and Series: Volume 3, More Special Functions}
        (Gordon and Breach, New York, 2002).

\bibitem{PRPAG-2018} L.~A.~Pach\'{o}n, A.~~Rela\~{n}o, B.~Peropadre, and A.~Aspuru-Guzik:
        Origin of the $1/f^\alpha$ spectral noise in chaotic and regular quantum systems.
        Phys. Rev. E {\bf 98}, 042213 (2018).

\bibitem{CCG-1985} G.~Casati, B.~V.~Chirikov, and I.~Guarneri:
        Energy-level statistics of integrable quantum systems.
        Phys. Rev. Lett. {\bf 54}, 1350 (1985).

\bibitem{RV-1998} M.~Robnik and G.~Veble:
        On spectral statistics of classically integrable systems.
        J. Phys. A: Math. Gen. {\bf 31}, 4669 (1998).

\bibitem{BH-1976} H.~P.~Baltes and E.~R.~Hilf: {\it Spectra of Finite Systems}
        (Bibliographisches lnstitut, Mannheim, 1976).

\bibitem{CK-1999} R.~D.~Connor and J.~P.~Keating:
        Degeneracy moments for the square billiard.
        J. Phys. G: Nucl. Part. Phys. {\bf 25}, 555 (1999).

\bibitem{Wong} R.~Wong: {\it Asymptotic Approximations of Integrals}, Vol. 34. SIAM, 2001.


\smallskip
\end{thebibliography}
\end{document}